%% file: main.tex
\documentclass[11pt]{article}
\usepackage{natbib}
\usepackage[utf8]{inputenc}
\usepackage[left=2cm,top=2cm,right=2cm,bottom=2cm]{geometry}
\usepackage{amsmath}
\usepackage[usenames,dvipsnames]{xcolor}
\usepackage{amsthm}
\usepackage{amssymb}
\usepackage[para]{footmisc}
\usepackage{color}
\usepackage{ae,aecompl}
\usepackage{inconsolata}
\usepackage{graphicx}
\usepackage[colorinlistoftodos]{todonotes}
\usepackage{hyperref}
\usepackage{xurl}
\usepackage{parskip}
\usepackage{mathtools}
\usepackage{tikz}
\usepackage{caption,subcaption}
\usepackage{enumitem}
\usepackage{dsfont}
\usepackage{bbm}
\usepackage{cancel}
\usepackage{commath}
\usepackage{MnSymbol}
\usepackage{mathdots}
\usepackage{algorithm}
\usepackage[noend]{algpseudocode}
\usepackage{pgfplots}
\usepackage{ucs} 
\usepackage{tkz-tab}
\usepackage{caption}
\usepackage{latexsym}
\usepackage{subcaption} 

\usepackage[toc,page]{appendix}
\usetikzlibrary{matrix,decorations.pathreplacing}
\usepackage{booktabs}
\usepackage{doi}

\algnewcommand{\algorithmicand}{\textbf{ and }}
\algnewcommand{\algorithmicor}{\textbf{ or }}
\algnewcommand{\OR}{\algorithmicor}
\algnewcommand{\AND}{\algorithmicand}

\newcommand{\defeq}{\vcentcolon=}
\newcommand{\eqdef}{=\vcentcolon}
\newcommand{\pluseq}{\mathrel{{+}{=}}}
\newcommand\undermat[2]{%
  \makebox[0pt][l]{$\smash{\underbrace{\phantom{%
    \begin{matrix}#2\end{matrix}}}_{\text{$#1$}}}$}#2}
\theoremstyle{plain}
\newtheorem{theorem}{Theorem}[section]

\newtheorem{lemma}[theorem]{Lemma}

\theoremstyle{definition}
\newtheorem{definition}[theorem]{Definition}

\newtheorem{example}[theorem]{Example}
\theoremstyle{remark}
\newtheorem{remark}[theorem]{Remark}

\numberwithin{equation}{section}

\DeclareMathOperator{\sign}{sign}

\algnewcommand{\IfThenElse}[3]{
  \State \algorithmicif\ #1\ \algorithmicthen\ #2\ \State \algorithmicelse\ #3}
\algnewcommand{\IfThen}[2]{
  \State \algorithmicif\ #1\ \algorithmicthen\ #2}

\algrenewcommand\algorithmicrequire{\textbf{Input:}}
\algrenewcommand\algorithmicensure{\textbf{Output:}}
\newcommand{\R}{\mathbb{R}}
\newcommand{\ev}{\mathbb{E}}
\DeclareMathOperator{\Var}{Var}
\DeclareMathOperator{\Corr}{Corr}
\DeclareMathOperator{\Cov}{Cov}

\newcommand{\ceil}[1]{\left\lceil #1 \right\rceil}
\newcommand{\floor}[1]{\left \lfloor #1 \right \rfloor}

\pgfplotsset{compat=1.17}

\newlength{\myimagewidth}
\begin{document}
\setlength{\myimagewidth}{\dimexpr\columnwidth/3-1em\relax}
\title{Simulation methods and error analysis for trawl processes and ambit fields}
	\author{Dan Leonte and Almut E. D. Veraart\\
	Department of Mathematics, Imperial College London}
	\maketitle
\noindent \begin{abstract}
\noindent Trawl processes are continuous-time, stationary and infinitely divisible processes which can describe a wide range of possible serial correlation patterns in data. In this paper, we introduce new simulation algorithms for trawl processes with monotonic trawl functions and establish their error bounds and convergence properties. We extensively analyse the computational complexity and practical implementation of these algorithms and discuss which one to use depending on the type of L\'evy basis. We extend the above methodology to the simulation of kernel-weighted, volatility modulated trawl processes and develop a new simulation algorithm for ambit fields. Finally, we discuss how simulation schemes previously described in the literature can be combined with our methods for decreased computational cost.\newline


\end{abstract}
\input{Introduction}
\input{The_Trawl_process_framework}
\input{Simulation_algorithms}
\input{Extensions_to_kernel_weighted_trawls}
\input{Extensions_to_ambit_field}
\input{proofs}
\input{appendix_file.tex}
\bibliographystyle{agsm}
\bibliography{bibs/bibliography_file}
\end{document}

%% file: Introduction.tex
\section{Introduction}
This paper introduces new simulation algorithms for trawl processes and ambit fields and establishes their error bounds and convergence properties. Deriving efficient simulation schemes, easily adaptable implementations and understanding the corresponding theoretical and numerical errors come naturally as first steps before employing such processes to model real-world data. 

Trawl processes and ambit fields have been introduced in the context of Ambit Stochastics, which was first developed to model physical phenomena such as turbulent flow and tumour growth by \cite{barndorff2007ambit}. Since its introduction, Ambit Stochastics proved to be a powerful modelling tool in other settings, such as spatio-temporal statistics \cite{nguyen_DG_RG_comparison}, brain imaging \cite{brain_imaging} and finance \cite{barndorff2014integer}. \cite{taqqu} studied trawl processes under the name of 'upstairs representations' and used them to model workloads for network communications. Together with the theoretical development of the topic, multiple computer libraries became available, as developed by \cite{git_valentin,R_package_almut_integer_trawl,Mathlab_gituhb_repo,git_nguyen_super_ou,git_nguyen_mixed_super_ou} and \cite{emil_github}. Apart from the extensive analysis of existing and new simulation schemes, we release a unified Python library at \cite{Leonte_Ambit_Stochastics_2022}, which contains efficient implementations of the discussed algorithms.

We start our analysis in the temporal setting with the class of trawl processes, then expand to the spatio-temporal case of ambit fields. Trawl processes are stationary and infinitely divisible stochastic processes which heavily rely on the following two concepts: the trawl set $A_t$, i.e.~the region which influences the value of the trawl process $X$ at time $t,$ and the L\'evy basis $L$, a type of random measure which extends the concept of noise from Gaussian and Poisson random measures to a general infinitely divisible setting. The trawl process $X$ is then defined as the L\'evy basis evaluated over the region of interest $X_t = L(A_t).$
This framework enforces the modelling belief that the value of the trawl process at time  $t$ is only influenced by a subset of the whole system, represented here by the trawl set $A_t,$ and not by the entire system. In many settings, the choice of the trawl set is inspired by the physical knowledge of the phenomenon to be modelled. A great advantage of trawl processes is the flexibility of the autocorrelation structure and of the marginal distribution, which can be chosen independently. Indeed, the areas of the overlaps $A_t \cap A_s$ determine the correlations $\Corr\left(X_t,X_s\right)$ and the L\'evy basis determines the marginal law of $X,$ which can be any infinitely divisible distribution.  This allows for the modelling of data displaying stylized facts, such as non-Gaussianity or heavy tails and offers a concrete and tractable alternative to modelling via stochastic partial differential equations (SPDEs), whose solutions can even be difficult to approximate numerically. A natural extension of trawl processes to spatio-temporal fields is given by random fields $Y$ defined as $Y(t,x) = Y_t(\mathbf{x}) = L\left(A_t(\mathbf{x})\right),$
where the set $A_t(\mathbf{x}) \subset \R^{d}$ now depends on both time and $d-1$ spatial coordinates. We call this a simple field. Simulation methods previously described in the literature and which are applicable to trawl processes include simulation via grid discretization in \cite{jonsdottir2008levy}, via compound Poisson processes in the case of integer-valued trawls in \cite{barndorff2014integer} and by using a slice partition in \cite[Chapter 4.4]{noven_phd_thesis}, \cite[Chapter 8.6]{ambit_book}. We generalize these methods, derive their theoretical errors and computational complexities and discuss practical implementation details. Further, we expand on the slice partition method, develop a novel simulation algorithm for simple ambit fields and explain how the calculations required for higher accuracy can be performed ahead of the simulations, amortising the computational time across simulations. This allows for the implementation of high-accuracy simulation studies and simulation-based inference.

Recent empirical work in areas such as environmental sciences in \cite{huang2011class} and energy pricing in \cite{benth_electricity,9966b470f5bd4eb5923297a0a31afb74} suggests the presence of volatility clusters, and hence of a stochastic volatility, which can be easily incorporated into the Ambit Stochastics framework. Indeed, we consider the volatility modulated, kernel-weighted trawl processes $X_t = \int_{A_t} K_t\left(\bar{t},\bar{\mathbf{x}}\right) \sigma(\bar{t})  \mathrm{d}L\left(\bar{t},\bar{\mathbf{x}}\right),$ and their spatio-temporal analogue, ambit fields $Y_t(\mathbf{x}) = \int_{A_t(\mathbf{x})} K_{t,\mathbf{x}}\left(\bar{t},\bar{\mathbf{x}}\right) \sigma\left(\bar{t},\bar{\mathbf{x}}\right)  \mathrm{d}L(\bar{t},\bar{\mathbf{x}}),$ where the deterministic kernel $K$ multiplied by the stochastic volatility $\sigma$ is integrated against the L\'evy basis $L.$ The integration is understood in the sense of \cite[Theorem 2.7]{rajput1989spectral} for deterministic $\sigma$ and in the sense of \cite{walsh,bichteler1983random} and \cite{kluppelberg} for stochastic $\sigma$. This general formulation introduces a kernel and a stochastic volatility with respect to the basis model of trawl processes and simple ambit fields and offers a complex framework which can be used to explicitly construct random fields with certain statistical properties, such as symmetry in space and time  \cite{barndorff2015intermittent}. We improve on the grid methods previously used in the literature by \cite{nguyen_DG_RG_comparison,emil_github} and show that despite the added terms, the trawl process simulation algorithms can be directly applied for the efficient simulation of kernel-weighted, volatility modulated trawl processes and ambit fields.
\subsection{Contributions of the paper}
We expand on the grid discretization algorithm and derive two new schemes for the simulation of trawl processes: the compound Poisson and the slice partition methods. Out of these schemes, only the slice partition gives exact simulation and accommodates any monotonic trawl shape and any infinitely divisible distribution described via the L\'evy basis, requiring only access to samples from the marginal distribution of the L\'evy basis. We derive the error bounds and convergence properties of the inexact algorithms and discuss the computational complexity of each algorithm.
We extend the slice partition method from trawl processes and develop a novel simulation algorithm for simple ambit fields. In general, this does not lead to exact simulation. However, it has the advantage that the calculations required for higher accuracy only need to be performed once, before the simulation, leading to amortized computational cost across simulations, as opposed to the grid and compound Poisson methods, which in general require an increased  cost per simulation; we also discuss in which situations the compound Poisson method might be preferable to the slice partition method. This allows for the practical implementation of high accuracy simulation studies. One direct application is parameter inference, where we simulate trawl processes or simple ambit fields and attempt to infer the shape of the ambit set and the parameters of the L\'evy seed. Further, in the context of simulation-based inference, bootstrap methods provide confidence intervals for the inferred parameter, in the settings of maximum likelihood or generalized method of moments estimation. Such studies have already been performed for integer-valued trawls and spatio-temporal Ornstein-Uhlenbeck processes in \cite{barndorff2014integer,nguyen_DG_RG_comparison}. 

Finally, motivated by the high computational complexity and relative inefficiency of grid methods, as discussed in Subsection \ref{subsection:summary_and_sampler}, we show how the compound Poisson and slice partition methods can be generalized to the simulation of volatility modulated, kernel-weighted trawl processes and ambit fields. We release a Python library containing the simulation algorithms discussed in the paper, see \cite{Leonte_Ambit_Stochastics_2022}.
\subsection{Structure of the paper}
Section \ref{section:The Trawl process framework}  defines the notion of L\'evy bases, which can be viewed as non-Gaussian extensions of Gaussian white noise, and settles the notation and theoretical framework. In particular, we discuss an extension of the L\'evy-Khintchine theorem from L\'evy processes and give formulae for the cumulant, autocorrelation structure and marginal distribution of the trawl process. Section \ref{section:simulation} presents and compares the three simulation algorithms for trawl processes, employing grid discretizations, compound Poisson processes and slice partitions and analyses the convergence properties and computational complexity, first in the case of bounded trawl sets, and then in the unbounded case. Sections \ref{section:extensions_to_VMKWTP} and \ref{section:extensions_to_ambit_field_simulation} extend the above methodologies to kernel-weighted, volatility modulated trawl processes and ambit fields. In particular, Subsection \ref{subsection:simple_ambit_fields} further derives a new simulation scheme for simple ambit fields via Monte Carlo methods. Proofs that have been omitted from the main body can be found in Section \ref{section:proofs}. Background material and a discussion of efficient implementations of discussed algorithms can be found in the Appendix. 

%% file: The_Trawl_process_framework.tex
\section{Trawl processes and their properties}
\label{section:The Trawl process framework}
We first introduce the notation and preliminaries needed in this section. For a set $S \subset  \R^d,$ let $\mathcal{B}_\text{Leb}(S)$ denote the collection of Borel measurable sets of finite Lebesgue measure which are contained in $S$. We view $S$ as a subset of space-time, where the first coordinate gives the time component and the last $d-1$ coordinates give the spatial component. We say that the measure $l$ is finite if $l(\R) < \infty$ and infinite otherwise. By a L\'evy measure $l$ on $\R$ we mean a (possibly infinite) Borel measure with $l({0}) =0$ and $\int_{\R} \min{(1,y^2)} l(\mathrm{d}y) < \infty.$ Finally, for a random variable $X$ we define the cumulant (log-characteristic) transform $C(\theta,X) = \log \left(\ev\left[e^{i \theta X}\right]\right)$ \citep[cf.][p. 33]{ken1999levy} and write $X\stackrel{d}{=}Y$ if $X$ and $Y$ have the same law. 

We formally define L\'evy bases and present some of their theoretical properties. We define the trawl process $X_t = L(A_t)$ as the L\'evy basis evaluated over a collection of sets of interest $A_t$ and discuss its marginal distribution and autocorrelation structure.
\subsection{L\'evy bases}
\begin{definition}[L\'evy basis]
\label{def:levy_basis}
A L\'evy basis $L$ on $S$ is a collection $\left\{L(A): A \in \mathcal{B}_\text{Leb}(S)\right\}$ of infinitely-divisible, real-valued random variables such that for any sequence $A_1,A_2,\ldots$ of disjoint sets in $\mathcal{B}_\text{Leb}(S),$ the random variables $L(A_1),L(A_2),\ldots$ are independent and further, if  $\cup_{j=1}^{\infty} A_{j} \in \mathcal{B}_\text{Leb}(S)$, then $L\left(\cup_{j=1}^{\infty} A_{j}\right)=\sum_{j=1}^{\infty} L\left(A_{j}\right)$ a.s.
\end{definition}
In the following, we assume that the L\'evy bases $L$ is homogeneous; a thorough discussion of this property can be found in Chapter 5.1 of \cite{ambit_book}.
\begin{definition}[Homogeneous L\'evy basis] A L\'evy basis $L$ on $S$ is homogeneous if there exist $\xi \in \R$, $a \in \R_{\ge 0}$ and a L\'evy measure $l$ on $\R$ such that for any $A \in \mathcal{B}_\text{Leb}(S),$ the following holds
\begin{equation*} 
  C(\theta,L(A))
  =  \left( i \theta \zeta-\frac{1}{2} \theta^{2} a+\int_{\mathbb{R}}\left(e^{i \theta y}-1-i \theta y \mathbf {1}_{[-1,1]}(y)\right) l(\mathrm{d} y) \right) \mathrm{Leb}(A).
\end{equation*}   \label{def:simplified_cumulant}
\end{definition}
Another important concept is that of the L\'evy seed, see e.g.~in \cite{barndorff2014integer}.
\begin{definition}[L\'evy seed]\label{def:levy_seed}
A random variable $L^{'}$ is called a L\'evy seed of the L\'evy basis $L$ if 
\begin{equation*}
    C(\theta,L^{'}) = i \theta \zeta-\frac{1}{2} \theta^{2} a+\int_{\mathbb{R}}\left(e^{i \theta y}-1-i \theta y \mathbf {1}_{[-1,1]}(y)\right) l(\mathrm{d} y).
\end{equation*} 
\end{definition}
Then
\begin{equation}
    C(\theta,L(A)) = \mathrm{Leb}(A) C(\theta,L^{'}), \label{eq:relate_cumul_L_with_cumul_levy_seed}
\end{equation}
and the distribution of the L\'evy seed determines the distribution of the L\'evy basis. Note that if $L^{'} \stackrel{d}= L^{''},$ then $L^{''}$ is also a L\'evy seed. Further, $L^{'}$ is infinitely divisible and to each L\'evy basis $L$ we can associate the L\'evy-Khintchine triplet $(\xi,\,a,\,l)$ of $L^{'}$, which fully characterises the distributional properties of $L$. In the above triplet, $\xi$ denotes the drift term, $a$ the variance of the Gaussian component and $l$ the L\'evy measure of the jump part \citep[cf.][p. 37]{ken1999levy}. Differentiating \eqref{eq:relate_cumul_L_with_cumul_levy_seed} once, respectively twice with respect to $\theta$ gives 
\begin{align}
    \ev \left[L(A) \right]    &=  \mathrm{Leb}(A) \,  \ev\left[L^{'}\right],     \label{eq:mean_L(A)}  \\
    \Var\left(L(A)\right)     &=  \mathrm{Leb}(A) \, \Var\left(L^{'}\right), \label{eq:var_L(A)}
\end{align}
and taking higher derivatives gives the relation between the cumulants of $L(A)$ and these of $L^{'}$.

Finally, to construct a trawl process, we need to choose the trawl sets. In the following, we restrict our attention to trawl processes with monotonic trawl functions, i.e.~when the trawl sets are of the form
\begin{equation*}
    A_t = A + (t,0), \qquad A = \{(s,x) \in \R^2 \colon  s < 0, 0 <  x < \phi(s) \},
\end{equation*}
where $\phi \colon (-\infty,0] \to \mathbb{R}_{\ge 0}$ is a 
continuous increasing 
function. 
Define the trawl process $X = \left(X\right)_{t \ge 0}$ by the L\'evy basis evaluated over the trawl set $X_t = L(A_t).$  
We note that, while the trawl process $X$ is defined to take values in $\R,$ the trawl set is chosen as a subset of $\R^2,$ i.e.~it includes an abstract spatial dimension in addition to the temporal dimension. Further, the trawl set is non-anticipative, in the sense that $A_t$ does not contain any points $(s,x)$ with $s>t.$ If there is some $T<0$ such that $\phi(T)=0,$ then $A$ is compactly supported and we say that the trawl is bounded. Otherwise, we say the trawl is unbounded. Generalizations are straightforward for $A \subset \R^{d}$ with $d > 2.$
\subsection{Marginal distribution}
\label{subsection_marginal_distr}
As seen in Definition \ref{def:levy_basis}, the only restriction on the marginal distribution of a trawl process is that it has to be infinitely divisible. This provides a rich class of stochastic processes supported on the integers, on the real line and on the positive or negative real line, with short or long memory and light or heavy tails. Some examples include the following processes.

\textbf{Integer-valued trawl processes}
\begin{example}[Poisson L\'evy basis]
Let $L'\sim \text{Poisson}(\nu)$ for some intensity $\nu >0.$ Then $X_t = L(A_t) \sim \text{Poisson}(\nu\mathrm{Leb}\left(A\right)).$
\end{example}
\begin{example}[Skellam L\'evy basis]
Let $L'\sim \text{Skellam}(\mu_1,\mu_2),$ i.e.~$L^{'}\sim N_1 - N_2$ with $N_1,N_2$ independent and Poisson distributed with intensities $\mu_1,\mu_2 >0.$  Then $X_t = L(A_t) \sim \text{Skellam}(\mu_1 \mathrm{Leb}\left(A\right),\mu_2 \mathrm{Leb}\left(A\right)).$
\end{example}
\textbf{Real valued trawl processes}
\begin{example}[Gaussian L\'evy basis]
Let $L^{'}\sim \mathcal{N}(\mu,\sigma^2)$ be Gaussian distributed with mean $\mu$ and variance $\sigma^2.$ Then $X_t = L(A_t) \sim \mathcal{N}\left(\mu \mathrm{Leb}\left(A\right), \sigma^2 \mathrm{Leb}\left(A\right)\right).$
\end{example}
\begin{example}[Cauchy L\'evy basis]
Let $L^{'} \sim \textrm{Cauchy}(\gamma)$ with scale parameter $\gamma >0.$ Then $X_t = L(A_t) \sim \textrm{Cauchy}(\gamma \mathrm{Leb}\left(A\right)).$
\end{example}
\textbf{Positive real valued trawl processes}
\begin{example}[Gamma L\'evy basis]
Let $L^{'} \sim \text{Gamma}(k,\theta)$ with shape and rate parameters $k,\theta >0$ and pdf
 $   p(x) = \frac{1}{\theta^k\Gamma(k)} x^{k-1} e^{- x/\theta}, \ x>0.$
Then $X_t = L(A_t) \sim \text{Gamma}(k \mathrm{Leb}\left(A\right),\theta).$
\end{example}
\begin{example}[Inverse Gaussian L\'evy basis]
Let $L'\sim \text{IG}(\delta,\gamma),$ with parameters $\delta,\gamma > 0$ and pdf 
 $p(x) = \frac{\delta}{\sqrt{2 \pi x^3}}e^{\delta \gamma - \frac{1}{2}\left(\frac{\delta^2}{x} +\gamma^2 x \right)}, \ x>0.$
Then $X_t = L(A_t) \sim \text{IG}(\delta \mathrm{Leb}\left(A\right), \gamma).$
\end{example}
A more general example is given by the class of trawl processes with stable distributions.
\begin{example}[Stable L\'evy basis]\label{example:stable}
Let $L^{'}\sim \mathrm{Stable}(\alpha,\, \beta,\, c,\, \mu)$ have a stable distribution with stability and skewness parameters $\alpha \in (0,2],\,\beta \in [-1,1]$ and location, scale parameters $\mu \in \R$, $c >0$, defined through the cumulant function $ C(\theta,L^{'}) = i \theta \mu + \abs{c\theta}^\alpha \left(1 - i \beta \sign{(\theta)} \Phi\right)$, where $\Phi = \tan{\frac{\pi \alpha}{2}}$ if $\alpha \neq 1$ and $-\frac{2}{\pi} \log{\abs{\theta}}$ if $\alpha =1$.
Then $X_t = L(A_t) \sim \textrm{Stable}(\alpha,\, \beta,\, c \textrm{Leb(A)}^{1/\alpha},\, \mu\mathrm{Leb(A)})$. The support of $L^{'}$ is $[\mu,\infty)$ if $\alpha <1,\,\beta=1$, $(-\infty,\mu]$ if $\alpha<1,\,\beta=-1$ and $\R$ otherwise.
\end{example}
L\'evy bases can thus be seen as a generalization of the Gaussian white noise process to a class of random measures with flexible marginal distributions.
\subsection{Covariance and Correlation structure}
\label{subsection: autocov}
We saw previously that the distribution of $L^{'},$ together with the Lebesgue measure of the trawl set, determines the marginal distribution of the trawl process $X_t.$ Similarly, the shape of the trawl set $A,$ specified here by the trawl function $\phi,$ determines the autocorrelation structure of $X_t.$ Indeed, note that $A_t \backslash A_{t+h},\, A_{t+h} \backslash A_t$ and $A_t \cap A_{t+h}$ are disjoint, hence the random variables
$L\left(A_t \backslash A_{t+h}\right),\,L\left(A_{t+h} \backslash A_t\right)$ and $L\left(A_t \cap A_{t+h}\right)$ are independent. By using this decomposition and by \eqref{eq:mean_L(A)} and \eqref{eq:var_L(A)}, we obtain that $\Cov(X_t,X_{t+h}) = \mathrm{Leb}\left(A\cap A_h\right) \, \Var(L^{'}) = \int_{-h}^0 \phi(s) \mathrm{d}s \Var\left(L^{'}\right) $ and further that
    \begin{equation}
    \rho(h)  \defeq \Corr(X_t,X_{t+h}) = \frac{\mathrm{Leb}\left(A \cap A_h\right)}{\mathrm{Leb}\left(A\right)} = \frac{\int_{-h}^0 \phi(s) \mathrm{d}s}{\int_{-\infty}^0 \phi(s) \mathrm{d}s}. \label{eq:autocor_intro}
\end{equation}
Thus we have a representation of the autocorrelation function $\rho$ solely in terms of the trawl function $\phi$. In particular, the trawl process can realize any positive, strictly decreasing autocorrelation function. Figure \ref{fig:trawl_process_simulation} displays some realisations of the trawl process, with short and long memory, light and heavy tails, simulated by the algorithm from Section \ref{subsection:slice_partition_algorithm}.
\begin{figure}[ht]
    \centering 
\begin{minipage}[t]{.33\textwidth}
\begin{subfigure}{\textwidth}
  \includegraphics[width=\linewidth]{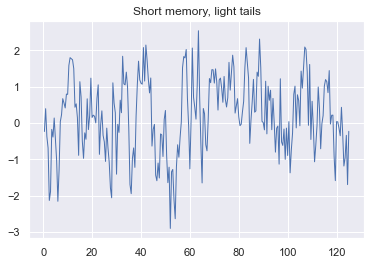}
  \caption{}
  \label{fig:short memory, light tails}
\end{subfigure}\hfil 
\begin{subfigure}{\textwidth}
  \includegraphics[width=\linewidth]{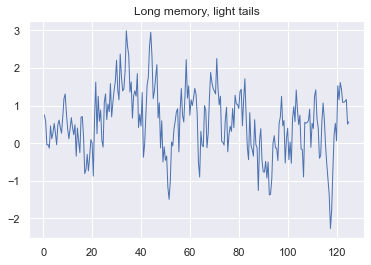}
  \caption{}
  \label{fig:long memory, light tails}
\end{subfigure}
\end{minipage}\hfil
\begin{minipage}[t]{.33\textwidth}
\begin{subfigure}{\textwidth}
  \includegraphics[width=\linewidth]{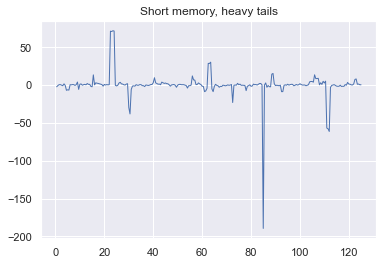}
  \caption{}
  \label{fig:short memory, heavy tails}
\end{subfigure}\hfil 
\begin{subfigure}{\textwidth}
  \includegraphics[width=\linewidth]{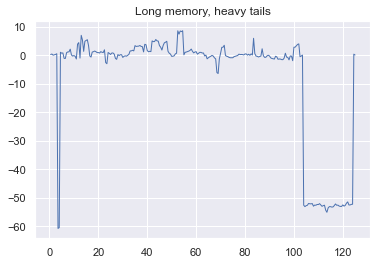}
  \caption{}
  \label{fig:long memory, heavy tails}
\end{subfigure}
\end{minipage}
\begin{minipage}[t]{.33\textwidth}
\begin{subfigure}{\textwidth}
  \includegraphics[width=\linewidth]{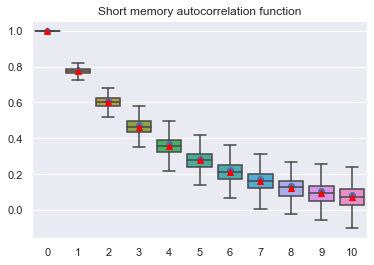}
  \caption{}
  \label{fig:short_memory_acf}
\end{subfigure}\hfil 
\begin{subfigure}{\textwidth}
  \includegraphics[width=\linewidth]{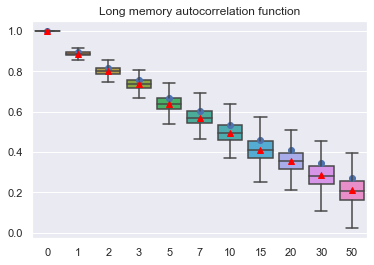}
  \caption{}
  \label{fig:long_memory_acf}
\end{subfigure}
\end{minipage}
        \caption{a-d) Realisations of the trawl process at times $\tau,\ldots,250 \tau$ with $\tau = 0.5,$ obtained by varying the trawl function and L\'evy seed. For short and long memory we set $\phi_{\textrm{sm}}$, $\phi_{\textrm{lm}} \colon (-\infty,0] \to \R_{\ge 0}$ given by $\phi_{\textrm{sm}}(t)= e^{t}$ and $\phi(t)_{\textrm{lm}} = 0.5(1-t)^{-1.5}$, which result in trawl sets of Lebesgue measure $1$. The corresponding autocorrelation functions $\rho_{\textrm{sm}}$, $\rho_{\textrm{lm}} \colon [0,\infty) \to \R_{\ge 0}$ are given by $\rho_{\textrm{sm}}(h) = e^{-h}$ for Figures \ref{fig:short memory, light tails} and \ref{fig:short memory, heavy tails}, and $\rho_{\textrm{lm}}(h) = (1+h)^{-0.5}$ for Figures \ref{fig:long memory, light tails} and \ref{fig:long memory, heavy tails}. For light and heavy tails, we set $L^{'}_{\textrm{lt}}\sim \mathcal{N}(0,1)$, respectively $L^{'}_{\textrm{ht}}\sim \textrm{Cauchy}(1)$. Since the Lebesgue measure of both trawl sets is $1$, the corresponding marginal distribution is $\mathcal{N}(0,1)$ in Figures \ref{fig:short memory, light tails} and \ref{fig:long memory, light tails} and $\textrm{Cauchy}(1)$ in Figures \ref{fig:short memory, heavy tails} and \ref{fig:long memory, heavy tails}. e-f) Theoretical and empirical autocorrelation functions of the trawl process, with trawl function $\phi_{\textrm{sm}}$ in e) and $\phi_{\textrm{lm}}$ in f). The boxplots describe the distribution of empirical autocorrelation functions at different lags, based on $500$ simulations, in which we simulate the trawl process at times $\tau,\ldots,1000\tau$ in e) and  $\tau,\ldots,5000\tau$ in f), with $\tau = 0.5.$ The blue circle and red triangle show the true, respectively the mean of the empirical autocorrelation functions. Note that the convergence of the empirical autocorrelation function, as a function of the number of simulated trawls, is much faster in the short memory case, where the two markers are superimposed.}
    \label{fig:trawl_process_simulation}
\end{figure}

%% file: Simulation_algorithms.tex
\section{Simulation algorithms for trawl processes and their convergence properties}
\label{section:simulation}
In this section, we present and compare three simulation algorithms for trawl processes: the grid discretization, the simulation via compound Poisson processes and the slice partition. Grid methods were previously considered in \cite{jonsdottir2008levy} and \cite{nguyen_DG_RG_comparison} for ambit field simulation, whereas simulation via compound Poisson processes was employed in \cite{barndorff2014integer} for the simulation of integer-valued trawls. The main disadvantage of these methods is that they are exact only for certain trawl shapes and marginal distributions of $L^{'};$ in general, the computational time increases as a function of the required accuracy. Based on the slice partition ideas from \citep[Chapter 8.6]{ambit_book} and \citep[Chapter 4.4]{noven_phd_thesis}, we describe the slice partition algorithm, which provides an efficient alternative for the exact simulation of monotonic trawls. We establish the convergence properties of these algorithms: in probability, in Skorokhod's topology and uniformly, providing MSE bounds on the theoretical error. In all three algorithms, we require samplers either from the law of $L(A)$ for sets $A$ of various Lebesgue measures or from the L\'evy measure $l$. We postpone discussing this technical but crucial requirement to Subsection \ref{subsection:summary_and_sampler}, when the presentation of the simulations schemes has finished and the need for such samplers is clear.

For ease of presentation, assume there is only one spatial component, i.e.~$d=2$. Consider a homogeneous L\'evy basis $L$ on $S= \R^2,$ with triplet $(\xi,\, a,\, l),$ where $\xi \in \R, a \in \R_{\ge 0}$ are constants and where $l$ is a L\'evy measure on $\R.$ Let the trawl set $A$ of finite Lebesgue measure be given by
\begin{equation*}
    A_t = A + (t,0), \qquad A = \{(s,x) \in \R^2 \colon  s < 0 ,\, 0 < x < \phi(s) \},
\end{equation*}
for some smooth, increasing function $\phi \colon \, (-\infty,0] \to \mathbb{R}_{\ge 0}$. We aim to simulate the trawl process $X_t= L(A_t) = L(A + (t,0))$ at equidistant times $\tau, \ldots,k \tau$. In the following, we simulate both the Gaussian and jump parts, but note that it is enough to simulate the jump part. Indeed, the covariance matrix $\Sigma$ of $L_g(A_\tau),\ldots,L_g(A_{k \tau})$ is given by the areas of the overlaps of the translated trawl sets. Thus the Gaussian part can be simulated by $Hx,$ where $H$ comes from the Cholesky decomposition of $\Sigma,$ i.e.~$\Sigma = H H^t,$ and $x$ is a vector sampled from the standard $k$ dimensional multivariate Gaussian. In general, the Cholesky factorisation has complexity $\mathcal{O}\left(k^3\right)$ and the matrix-vector multiplication $\mathcal{O}\left(k^2\right).$ 
\subsection{Algorithm I: grid discretization}
\label{subsection:grid_algo}
Assume that $A$ is bounded, i.e.~there exists $T<0$ such that $\phi(T)=0$ and $A = \{(t,x) \in  \R^2  : T < t \le 0, 0 < x < \phi(t)\};$ we later relax this assumption. Assume that $T+\tau<0;$ otherwise, the sets $A_\tau,\ldots,A_{k\tau}$ are disjoint and we can simulate the trawl process by drawing $k$ iid samples from the law of $L(A).$ In this algorithm, we discretize the rectangle $[T+\tau,k\tau] \times [0,\phi(0)]$ into a grid of cells, simulate the L\'evy basis over these cells and approximate $L(A)$ by $\sum{L(c)},$ where the sum is over cells $c$ which are fully contained in $A.$ 

Let the discretization step-sizes on the time, respectively space axes be $\Delta_t = \tau/N_t$, $\Delta_x = \phi(0) / N_x$ for some positive integers $N_t$,$N_x$. If $T/\Delta_t$ is not an integer, replace $T$ with $\floor{T/\Delta_t}\Delta_t$, where $\floor{\cdot}$ is the floor function and let $N = -T/\Delta_t$; this choice excludes boundary effects and ensures that all the cells we consider in this algorithm have equal area. In total, there are $\left((k-1)N_t + N\right)N_x$ cells contained in the grid on $[T+\tau,k\tau] \times [0,\phi(0)]$. Let $Y$ be an $N_x \times \left((k-1)N_t+N\right)$ random matrix with iid entries $Y_{ij}\stackrel{d}{=}L(c),$ corresponding to the L\'evy basis simulated over all the grid cells. Practical experiments show that for small values of $\Delta_t$ and $\Delta_x$, it is not feasible to hold a realisation of $Y$ in memory. Nevertheless, note that to simulate $L(A_{t}),$ it is enough to hold the samples $L(c)$ for cells $c$ contained in $[t+T ,t] \times [0,\phi(0)]$ in memory. To this end, for each trawl set $A_{l\tau,}$ define the corresponding set of $N\cdot N_x$ grid cells on $[l\tau+T,l\tau]\times [0,\phi(0)]$ by $G_l =\{g^l_{ij} : 1 \le i \le N_x, 1 \le j \le N \}$(see Figure \ref{fig:included_cells}), where each cell is of the form $g^l_{ij} = [l\tau +T+ (j-1) \Delta_t,\, l\tau + T+ j \Delta_t] \times [(i-1)\Delta_x,\,i \Delta_x]$. Define $I$ to be the $N_x\times N$ indicator matrix with entries $I_{ij}=1$ if $g^l_{ij} \in A_{l\tau}$ and $0$ otherwise and note that $I$ does not depend on the chosen trawl set $A_{l\tau}.$ Finally, let $Y_l = Y[:,\,(l-1)N_t+1:(l-1)N_t+N]$ be the $N_x\times N$ random matrix obtained by subsetting only the columns of $Y$ corresponding to cells in $G_l;$ the grid approximation of $L(A_{l\tau})$ is then given by $Y_l \odot I,$ i.e.~the sum of the entries of the component-wise product of matrices $Y_l$ and $I.$ Iteratively, at step $l+1,$ we can compute $Y_{l+1}$ from $Y_l$ by discarding $\{L(c)\}_{g \in G_l \backslash G_{l+1} }$ and adding new samples $\{L(c)\}_{g \in G_{l+1} \backslash G_l}$(see Figure \ref{fig:grid_update}). This corresponds to removing the first $N_t$ columns from the left of $Y_{l},$ adding $N_t$ new sampled columns to the right of $Y_{l}$ and approximating $L\left(A_{(l+1)\tau}\right)$ by $Y_{l+1} \odot I.$ The full procedure is given in Algorithm \ref{algo:grid_discretisation}.

\input{algorithms/grid_discretisation}
\begin{remark}
\label{remark_cell_and_lebesgue_decomp}
Algorithm \ref{algo:grid_discretisation} requires checking if a cell $c$ is fully contained in a monotonic trawl set $A_t$ (step $5$) and a sampler for $L(c)$ (steps $12,17$). For the first requirement, note that a cell $c$ is fully contained in a monotonic trawl set $A_t$ iff the `top-left' corner $(s,x)$ of the cell is in $A_t$. The cells can then be represented by their 'top-left' corner and the inclusion condition $c \subset A_t$ is equivalent to $ x \le \phi(s-t),$ which can be easily checked on a computer. The second requirement is discussed extensively in Subsection \ref{subsection:summary_and_sampler}. 
\end{remark}
\input{pictures/tikz_files/grid_approx_and_update}
Even if $A$ is unbounded, we can still apply Algorithm \ref{algo:grid_discretisation} by truncating and approximating $L(A)$ by $L(A \cap \{t>T\}).$ In this case, the algorithm has two sources of error: firstly, taking a grid discretization and considering a cell as part of the trawl if its 'top-left' corner is in the trawl set, and secondly, neglecting $L\left(A \cap \{t<T\}\right).$ We now derive the convergence properties of the grid discretization algorithm in the setting of both infill and increasing domain asymptotics; further, if the L\'evy seed $L^{'}$ has finite variance, we provide MSE bounds on the error. 
\begin{theorem}
\label{grid_discretisation_thm}
Let $A = \{(s,x) \in \R^2 \colon  s < 0,\, 0 < x < \phi(s) \}$ be a monotonic trawl set of finite Lebesgue measure, $T<0$ and $\Delta = (\Delta_t,\, \Delta_x)$ with $\Delta_t,\,\Delta_x>0$. Let $G_{T,\Delta}$ be the set of grid cells on $[T,0] \times [0,\phi(0)]$ with step-sizes $\Delta_t,\Delta_x$ on the time, respectively space axes. Let $L_{T,\Delta}(A)$ be the approximation of $L(A)$ with respect to 
$G_{T,\Delta}$ (see Figure \ref{fig:theorem3.3})
\begin{equation*}
    L_{T,\Delta}(A) = \sum_{g \in G_{T,\Delta} \colon \, g \subset A} L(g).
\end{equation*}
Let $T_n \to -\infty, \Delta_n = (\Delta^n_t,\Delta^n_x) \to 0$. Then $L_{T_n,\Delta_n}(A) \to L(A)$ in probability and, if $\Var{L^{'}}$ is finite, in $\mathcal{L}^2$; further, if $\Var{L^{'}}$ is finite, then $ \ev \left[\left(L_{T,\Delta}(A) - L(A)\right)^2 \right] \le  C_{T,\Delta}^2 \ev\left[L^{'}\right] + C_{T,\Delta} \Var{L^{'}},$
    where $C_{T,\Delta} = T \Delta_x + \phi(0) \Delta_x + \int_{-\infty}^{T} \phi(s) ds .$ In particular, if $A$ is bounded and $\Delta_t=\Delta_x$, then $
    \ev \left[\left(L_{T,\Delta}(A)- L(A)\right)^2 \right] =   \mathcal{O}\left(\Delta_t^2\right).$
\end{theorem}
\begin{figure}
    \centering
  \includegraphics{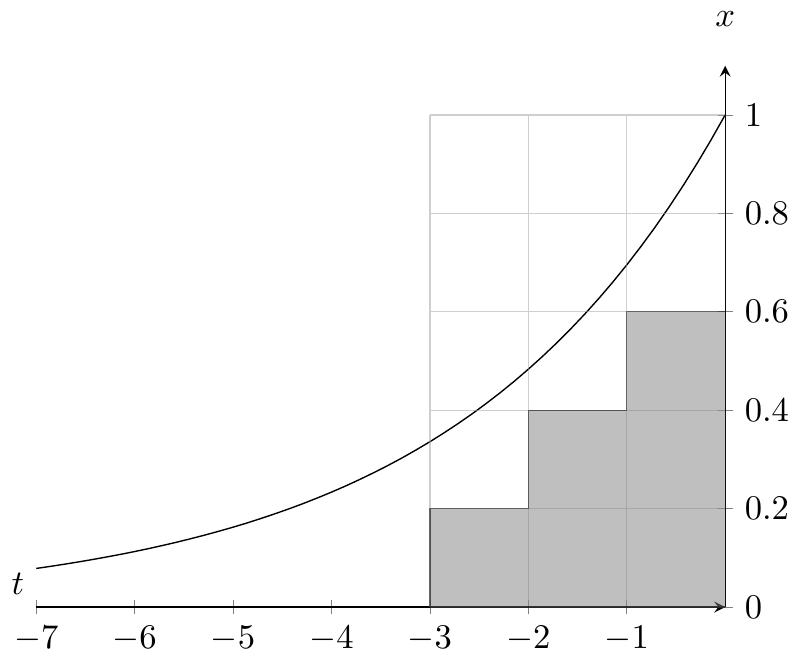}
  \caption{Illustration of the grid discretization for the unbounded trawl set $A = \{(t,x) \in \R^2 \colon  t < 0 , 0 < x < e^{2.75 t} \}$, with truncation parameter $T = -3$ and step-sizes $\Delta_t =1$, $\Delta_x = 0.2.$ The cells included in the simulation are shaded in gray. For unbounded trawls, we incur both truncation and discretization error.}
    \label{fig:theorem3.3}
\end{figure}
The MSE bound in Theorem \ref{grid_discretisation_thm} is not sharp for unbounded trawls, in the sense that it gives $\mathcal{L}^2$ convergence only when $T_n  \Delta_x^n \to 0$ (see proof in Section \ref{section:proofs}). Nevertheless, it provides a way to choose the truncation parameter $T$ and step-sizes $\Delta_t,\,\Delta_x$ for the purpose of computer simulations. Despite the convergence properties, Algorithm \ref{algo:grid_discretisation} is computationally expensive, requiring samples for $N_x(\left(k-1\right)N_t+N) =\left(\left(k-1\right)\tau + T \right) \cdot \phi(0)/ (\Delta_t \Delta_x)$ cells. In general, if $A \subset \R^{d}$ with $d>2$ and with same step-size $\Delta_t$ for all dimensions, the number of cells grows like $\mathcal{O}(1/\Delta_t^d)$ and holding the grids $G_l$ in memory, even one at a time, is not feasible. In the following, we discuss two alternatives: the compound Poisson and the slice partition algorithms, which are less computationally intensive. 
\subsection{Algorithm II: compound Poisson simulation}
\label{subsection: cpp}
Given a L\'evy basis $L$ with characteristic triplet $(\xi,\,a,\,l),$ by the first part of Theorem \ref{thm:levy_ito_decomp_for_levy_bases}, there exist a decomposition of $L$ into a Gaussian L\'evy basis $L_g,$ with triplet $(\xi,\,a,\,0),$ and a jump Levy basis $L_j,$ with triplet $(0,\,0,\,l)$, such that $L_g$ and $L_j$ are independent and $L=L_g+L_j.$ Assume that the Gaussian part has already been simulated, which can be done exactly by Algorithm \ref{algo:slice_partition_trawl_process_convolution} of the next subsection, or by using a Cholesky decomposition and a standard normal sampler. By the second part of Theorem \ref{thm:levy_ito_decomp_for_levy_bases}, there exists a Poisson random measure $N$ on $\mathcal{B}_{\mathrm{Leb}} (\R^2)\times \R$ with intensity measure $\nu \defeq \mathrm{Leb} \otimes l$ such that
\begin{equation*}
    L_j\left(A\right) = \int_A \int_{(-1,1)} y \ \mathrm{d} (N- \nu)(\mathbf{z},y) + \int_A \int_{\R \backslash(-1,1)} y \ \mathrm{d} N(\mathbf{z},y),
\end{equation*}
where $\mathrm{Leb}$ is the Lebesgue measure on $\R^2$, $\mathbf{z} = (t,x) \in \R^2$ and $y \in \R.$ In the above, $(t,x)$ are the time and space coordinates at which a jump appears, whereas $y$ is the value of the jump. For ease of presentation, assume that the trawl set is bounded; we relax this assumption at the end of the subsection. We analyse the cases of finite and infinite L\'evy measures separately.

If $l$ is a finite measure, i.e.~if $c \defeq l(\R) < \infty,$ we can simplify the above expression to
\begin{equation}
       L_j\left(A\right) =  \int_A \int_{\R \backslash(-1,1)} y \ \mathrm{d} N(\mathbf{z},y) - \mathrm{Leb}(A) \int_{-1}^1 y \ l(dy).
\end{equation}
With this representation in mind, to simulate the trawl process, we can simulate a Poisson point process on $\R^2$ with constant intensity $c,$ where to each generated point $(t_i,x_i)$ we associate a corresponding jump sample $y_i$ with law  $\tilde{l} \defeq l/c$. Let $\tilde{\xi} = \int_{-1}^1 y \ l(\mathrm{d}y)$. The value of $L_j(A)$ is then given by the difference between the sum of all jumps $y_i$ corresponding to points $\mathbf{z}_i = (t_i,x_i)$ contained in the trawl set $A$ and the drift term $\tilde{\xi} \mathrm{Leb}(A)$
\begin{equation*}
   L_j(A) = \sum_{i \colon  \ (t_i,x_i) \in A} y_i - \tilde{\xi}  \mathrm{Leb}(A) .
\end{equation*}
It follows that $L_j(A)$ is compound Poisson distributed, with jumps distributed according to $\tilde{l},$ minus a constant. This property is also clear from the simplified expression for the cumulant transform from Definition \eqref{def:simplified_cumulant}  
\begin{equation*}
      C(\theta,L_j(A)) =  c \mathrm{Leb}(A) \int_{\mathbb{R}}\left(e^{i \theta y}-1\right) \tilde{l}(\mathrm{d} y) - i\theta \tilde{\xi}  \mathrm{Leb}(A) .
\end{equation*}
Finally, as the trawl is bounded, we can choose $T<0$ be such that $\phi(T)=0$ and note that we only need to simulate the Poisson point process on $[T+\tau,k\tau] \times [0,\phi(0)].$ The full procedure for the simulation of trawl processes with bounded trawl sets and finite L\'evy measures is given in Algorithm \ref{algo:CPP}.
\begin{algorithm}
  \caption{Compound Poisson simulation}  \label{algo:CPP}
  \begin{algorithmic}[1]
      \Require Trawl function $\phi$ and $T<0$ with $\phi(T)=0;$ sampler $S(n)$ which returns an $n$ dimensional vector of iid samples from $\tilde{l};\, c= l(\R)$ and $\tilde{\xi} = \int_{-1}^1 y \ l(\mathrm{d}y)$; number of trawls to be simulated $k$ and distance $\tau$ between trawl sets; sampler $U(a,\,b,\,n)$ which returns an $n$ dimensional vector of iid samples from the uniform distribution on $[a,b]$
    \Ensure Vector $X$ containing the simulated values of the trawl process at times $\tau,\ldots,k\tau.$
      \Function{CppSimulation}{$S,\,c,\,\tilde{\xi},\,\phi,\,T,\,k,\,\tau$}
      \State $\nu \gets c  \phi(0) \left((k-1) \tau -T\right)$
      \State $N \sim Poisson(\nu)$
      \State $t \sim U\left(\left[T+\tau, k \tau\right]\right)$\Comment{Generate $N$ jump time samples} 
      \State $x \sim U\left(\left[0,\phi(0)\right]\right)$ \Comment{Generate $N$ jump height samples}
      \State $y \sim S(N)$ \Comment{Generate $N$ jump value samples}
      \State $X \gets \tilde{\xi}\mathrm{Leb(A)} \cdot \textrm{ones}(k)$
      \For{$i \in \{1,\ldots,N\}$}
        \For{$l \in \{1,\ldots,k\}$}
        \If{$x[i] < \phi(t[i]-l\tau)$} \Comment{Check if $(t_i,x_i) \in A_{l\tau}$}
        \State{$X[l] += y[i]$}
        \EndIf
      \EndFor
          \EndFor
      \Return $X$
      \EndFunction
  \end{algorithmic}
\end{algorithm}
Note that the nested for loops in steps $9$ and $10$ require an expected number of $\mathcal{O}(k^2)$ comparisons of the form $x[i] < \phi(t[i]-l\tau).$ Since $A$ is bounded, at most $\ceil{-T/\tau}$ consecutive trawl sets have non empty intersections, hence it is enough to do $\ceil{-T/\tau}$ comparisons in step $10$ and the complexity reduces to $\mathcal{O}(k).$

On the other hand, if $l$ is an infinite measure, we can no longer represent $L_j$ via a compound Poisson process; we need to truncate and discard jumps with magnitude below some threshold, as we would when simulating a L\'evy process. Let $l^\epsilon$ be the restriction of $l$ to $ (-\infty,-\epsilon) \cap (\epsilon,\infty)$ and $\tilde{l^\epsilon} \defeq {l^\epsilon} /{l^\epsilon\left(\R\right)}. $ Define the corresponding truncated L\'evy basis $L^\epsilon_j$
\begin{equation}
    L^\epsilon_j\left(A\right) = \int_A \int_{(-1,-\epsilon) \cup (\epsilon,1)} y \ \mathrm{d} (N- \nu)(\mathbf{z},y) + \int_A \int_{\R \backslash(-1,1)} y \ \mathrm{d} N(\mathbf{z},y),
    \label{eq:truncated_levy_basis}
\end{equation}
the resulting approximation $X_t^\epsilon = L_j^\epsilon(A_t)$ and $X_t = L_j(A_t)$. Note that $X^\epsilon$ can be simulated with Algorithm II and that
\begin{equation*}
    \ev\left[\left(L_j(A) - L^\epsilon_j(A)\right)^2\right]= \ev \left[\left(\int_A \int_{(-1,-\epsilon) \cup (\epsilon,1)} y \ \mathrm{d} (N- \nu)(\mathbf{z},y)\right)^2\right] =  \mathrm{Leb(A)}\int_{-\epsilon}^\epsilon y^2 \mathrm{d}l(y) \to 0 \text{ as } \epsilon \to 0,
\end{equation*}
which gives different convergence rates for different L\'evy measures $l$. Even when the L\'evy seed does not have any finite moments, the error $L_j(A) - L_j^\epsilon(A)$ is square integrable and converges to $0$ in $\mathcal{L}^2,$ which already improves on the convergence of Theorem \ref{grid_discretisation_thm}. We further establish convergence of $X^\epsilon$ as a stochastic process, rather than just at the level of the marginals $X_t.$
\begin{theorem}
\label{theorem:uniform_convergence}
The process $X^\epsilon$ converges uniformly on compacts to $X$ on the space of c\`adl\`ag paths as $\epsilon \to 0.$
\end{theorem}
We now relax the bounded trawl set assumption. Previous simulation methods for integer-valued trawl processes with unbounded trawls from \cite{ambit_book,VERAART2019110} approximated $L(A)$ by $L(A \cap \{t > T\})$ for some truncation threshold $T,$ as $L(A \cap \{t \le T\}) \to 0$ in probability as $T \to -\infty$. Apart from losing the uniform convergence properties, $T$ must be chosen with a very large absolute value, significantly increasing the computational time. For example, if the trawl process has long memory, with  autocorrelation function given by $\rho(t) = (1+t)^{1/4},$ then even choosing $T=-10^4$ would incur an unacceptable trawl set truncation error. The apparent difficulty comes from simulating a Poisson point process on an unbounded domain of finite area, which we will discuss next. 

In the case of unbounded trawl sets, apply Algorithm \ref{algo:CPP} to simulate $L_j^\epsilon$ on $[\tau,k\tau] \times [0,\phi(0)]$ and sample the number of atoms $\mathbf{z}_i = (t_i,x_i)$ of  $L_j^\epsilon$ on $A_\tau$ $N\sim\text{Poisson}(l^\epsilon(\R) \textrm{Leb}(A_\tau))$. Conditionally on $N$, the $t_i$'s are iid with density $p(t) = \phi(t-\tau) / \int_{-\infty}^0 \phi(s)ds$ for $t<\tau$ and $0$ otherwise. Finally, conditionally on $N$ and $t_i$, $x_i \sim U(0,\phi(t_i))$ and we can continue with Steps $6-11$ of Algorithm \ref{algo:CPP}. By decomposing $\phi$ into its convex and concave parts, we can draw samples with density $p$ by rejection sampling, as described in \cite{ccarj}.   

Finally, note that the proof of Theorem \ref{theorem:uniform_convergence} does not require a bounded trawl set, thus we obtain the same convergence results for unbounded trawls. Nevertheless, for infinite L\'evy measures, $l\left(\R \backslash(-\epsilon,\epsilon)\right) \to \infty$  as $\epsilon \to 0,$ hence the intensity $\nu$ of the Poisson process $N$ to be sampled in Step $3$ of Algorithm \ref{algo:CPP} diverges, leading to an increasing cost per simulation. In the next subsection, 
we present an algorithm which simulates trawl processes exactly, regardless of the type of L\'evy measure.
\subsection{Algorithm III: slice partition}
\label{subsection:slice_partition_algorithm}
In this algorithm we decompose the sets $A_\tau,\ldots,A_{k\tau}$ into a collection  $\mathcal{S}$ of disjoint slices $S$, simulate the values of the L\'evy basis $L$ over each slice and then set
\begin{equation*}
    X_{l\tau} = \sum_{S \subset A_{l\tau}} L(S).
\end{equation*}
Indeed, because of the indepedent-scatteredness of the L\'evy basis, i.e.~the first property in Definition \ref{def:levy_basis}, we can sample $\{L(S)\}_{S \in \mathcal{S}}$ independently; then, by the additivity of the L\'evy basis, i.e.~the second property in Definition \ref{def:levy_basis}, we can reconstruct the value of the trawl $X_{l\tau}$ by summing the values corresponding to the L\'evy basis simulated over the slices contained in $A_{l\tau}.$ In general, there could be up to $2^k$ slices of the form $B_1 \cap \ldots \cap B_k,$ where $B_l \in \left\{A_{l\tau},A_{l\tau}^C\right\}.$ For monotonic bounded trawls, the number of slices is $\mathcal{O}\left(k\right),$ whereas for monotonic unbounded trawls, it is $\mathcal{O}\left(k^2\right)$, making the simulation scheme feasible. We analyse the slice partition separately in the bounded and unbounded case. For ease of notation and without risk of confusion, we write $A_l$ for $A_{l\tau}.$

If there is some $T<0$ such that $\phi(T)=0,$ let $I = \ceil{-T/\tau}$, where $\ceil{\cdot}$ is the ceiling function and define the slice partition (see Figure \ref{fig:Ng2})
\begin{align*}
    S_{i1}  &=   \{t \leq \tau\}  \cap \left(A_i \backslash A_{i+1} \right),\\ 
    S_{ij} &=  \left(\{(j-1)\cdot \tau < t \leq j\cdot \tau\} \cap A_{i+j-1} \right) \backslash  A_{i+j},   \text{ for } j \ge 2. 
\end{align*}
Then exactly $I$ consecutive trawl sets have non-empty intersection and each of the trawl sets $A_{l\tau}$ contains exactly $I$ slices, making up for a total of $kI$ slices. 
Let $s_{ij} = \mathrm{Leb}\left(S_{ij}\right);$ by the translation invariance of the Lebesgue measure, $s_{ij} = s_{ij'}$ for $j,\,j' \ge 2.$ Hence to determine the areas of the slices $S_{ij}$ it is enough to compute $s_{ij}$ for $i \in \{1,\ldots,k\}$ and $j \in \{1,2\}$. A short calculation shows that 
\begin{align}
    s_{i1} &= \int_{-i\cdot \tau}^{(-i+1)\cdot \tau} \phi(t) dt, \label{trawl_1_areas}\\
    s_{i2} &= s_{i,1}-s_{i+1,1}, \label{trawl_2_areas}
\end{align}
where we set $s_{I+1,1} =0$. The above two equations fully specify the areas $s_{ij}$. Let $Y$ be the $I \times (I-1+k)$ random matrix of independent random variables $L\left(S_{ij}\right), 1 \le i \le I, \ 1 \le j \le k,$ padded with $I-1$ columns of $0$'s to the left and $F$ be the $I\times I$ lower diagonal matrix filer 
\input{pictures/slice/Y_tilde_compact}
The values of the trawl process at times $\tau,\ldots,k\tau$ are given by the convolution $Y*F$ (see Appendix \ref{def:convolution} for the definition of matrix convolution). The full procedure is given in Algorithm \ref{algo:slice_partition_trawl_process_convolution}. Since $s_{ij} = s_{ij'}$ for $j,\,j' \ge 2,$ we can vectorize steps $6-8$ by sampling columns of $Y.$ Further, the convolution step $8$ can be computed efficiently by taking advantage of the form of the filter $F$ (see Algorithm \ref{algo:slice_partition_trawl_process_implementation} in  Appendix \ref{appendix:B}), which is also implemented in \cite{Leonte_Ambit_Stochastics_2022}. In both cases, the number of operations is $\mathcal{O}(k).$ 
\input{algorithms/slice_partition_algo}

If the trawl set $A$ is unbounded, define the slice partition
\begin{equation*}
    S_{ij} = \left(A_j \cap A_{i+j-1}\right) \backslash A_{i+j},\ 1\le i,j \le k, \ i+j \le k+1. 
\end{equation*}
In total, we have $k(k+1)/2$ slices (see Figure \ref{fig:Ng1}). Algorithm \ref{algo:slice_partition_trawl_process_convolution} still applies, with the mention that the areas $s_{ij}$ have different formulae (see Equation \ref{areas_s_ij_non_compact_ambit_set} in the appendix), $Y$ is $k \times (2k-1),$ F is $k\times k$
\input{pictures/slice/Y_tilde_non_compact}

and we now perform $\mathcal{O}\left(k^2\right)$ operations. In practice, some of these $k(k+1)/2$ slices might have areas below machine precision, and for the purpose of computer simulations, we can approximate $Y * F$ by $Y\mathrm{[1:n,:]} * F\mathrm{[1:n,:]}$, where $Y\mathrm{[1:n,:]}$ and $F\mathrm{[1:n,:]}$ are obtained from $Y$ and $F$ by discarding the last $k-n$ rows, where $n \in \{1,\ldots,k-1\}$. In this case, the first $k-n$ trawls are not simulated exactly and the last $n$ trawls are simulated exactly; further, since the errors $\{\epsilon^n_l\}_{l=1}^{k-n}$ are given by
\begin{equation*}
    \epsilon^n_{l} =  L\left(\bigcup_{\substack{1 \le j \le l \\ i> n ,i +j \ge l+1}} S_{ij}\right),
\end{equation*}
 we have access to the joint distribution of the errors, and in particular, to the mean and variance of the errors. This can be used to calibrate the truncation parameter $n.$ 
 \begin{remark}
 Note that Theorem \ref{theorem:uniform_convergence} of the previous subsection establishes not just convergence of $X^{\epsilon}$ to $X$ at discretely observed times $\tau,\ldots,k\tau,$ but convergence of stochastic processes in the supremum norm. A similar result holds for the slice partition algorithm. To this end, simulate $X$ via the slice partition method at discrete times $\mathcal{D}_n = \{0,\frac{1}{2^n},\ldots,1\}$ and define $X^n(t) = X\left(\frac{\floor{2^n t}}{2^n}\right),$ where $\floor{\cdot}$ is the floor function.
 \end{remark}
\begin{theorem} The sequence of stochastic processes $X^n$ converges a.s. to $X$ in Skorokhod's J1 topology.  
 \label{theorem:skorohod}
\end{theorem} 
\begin{figure}[!tbp]
  \centering
  \subfloat[Illustrates the slice partition $S_{ij}$ of a bounded monotonic trawl, with $I=3$ and $k=4.$ Each of the trawl sets are decomposed into $I$ slices.]{\includegraphics[height = 5.75cm,keepaspectratio ]{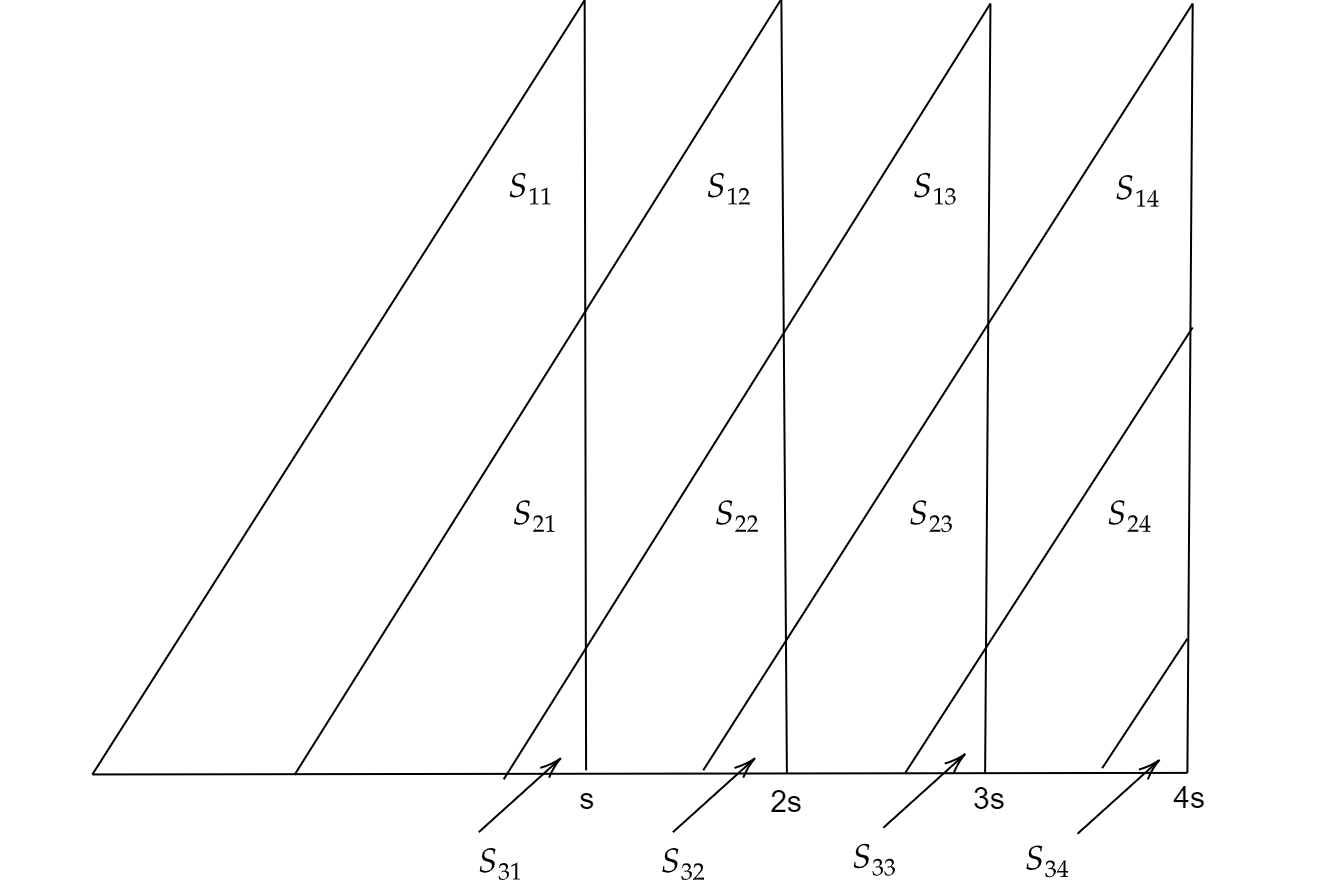}\label{fig:Ng2}}
  \hfill
\subfloat[Illustrates the slice partition $S_{ij}$ of monotonic trawl process with infinite decorrelation time and with $k=4.$ For each $i,$ there are $k-i+1$ slices $S_{ij}.$ ]{ \includegraphics[height = 5.75cm,keepaspectratio ]{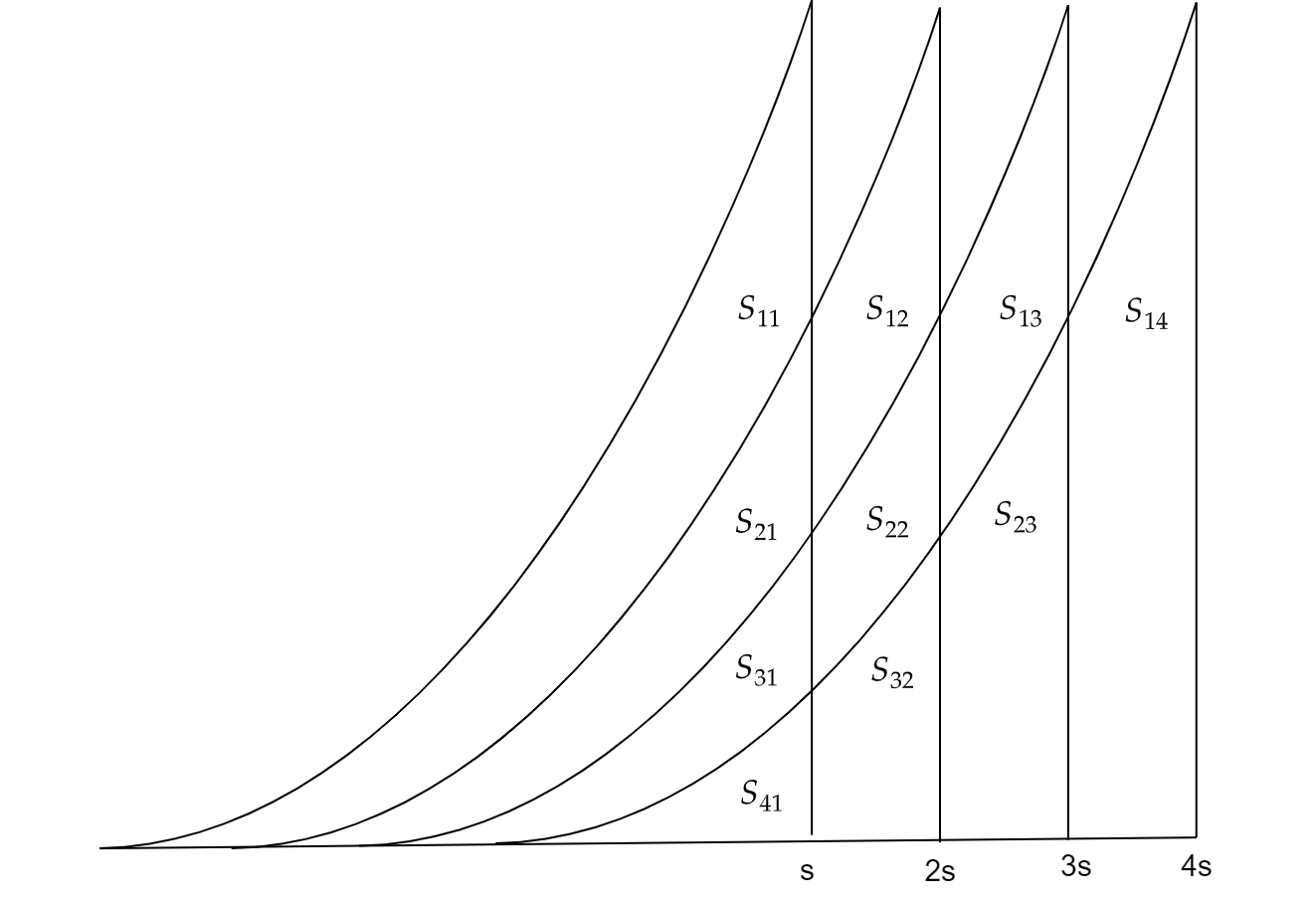}\label{fig:Ng1}} 
 \caption{The slice partition for monotonic trawls with (a) finite decorrelation time and (b) infinite decorrelation time }\label{fig:slice_partition}
\end{figure}

\subsection{Summary of convergence properties and discussion of algorithm requirements}\label{subsection:summary_and_sampler}
In Subsections \ref{subsection:grid_algo}-\ref{subsection:slice_partition_algorithm} we presented the grid discretization, compound Poisson and slice partition algorithms for the simulation of trawl processes; we derived their computational complexity and convergence properties and further discussed efficient implementation methods. Based on this analysis, we can compare the schemes and discuss in which cases one should be used over the other.

The grid discretization method is the most general one and can be employed to approximately simulate not only trawl processes, but also volatility-modulated, kernel weighted trawl processes and ambit fields, as defined in Sections \ref{section:extensions_to_VMKWTP} and \ref{section:extensions_to_ambit_field_simulation}. Nevertheless, it does not share the same convergence properties as the other two algorithms and further, it is the only algorithm which affects the autocorrelation function of the simulated trawl process. To quantify the discretization error by statistical measures, let $\rho_{T,\Delta}(h) \defeq \Corr{\left(L_{T,\Delta}\left(A_t\right),L_{T,\Delta}\left(A_{t+h}\right)\right)}$ be the theoretical autocorrelation function of the trawl process simulated by the grid method with truncation parameter $T$ and step-size $\Delta=\Delta_t=\Delta_x$. Let $\Var{\left(L_{T,\Delta}\left(A\right)\right)}$ be the corresponding variance. Figure \ref{fig:acf_and_var} shows slow convergence of $\rho_{T,\Delta}(h)$ to $\rho(h)$ and of $\Var{\left(L_{T,\Delta}\left(A\right)\right)}$ to $\Var\left(L(A)\right)$ for $A = \{(t,x) \in \R^2 \colon  t < 0 , 0 < x < \phi(t) \}$ with $\phi \colon (-\infty,0] \to \R_{\ge 0}$ given by $\phi(t) = 0.5(1-t)^{-1.5}.$ The simulation error in the finite scale (non-asymptotic) regime depends heavily on the rate of decay of the trawl function $\phi.$ Consequently, accurately simulating long memory trawl processes, such as the ones with $\phi_H(t) \propto (1-t)^{-H}$, $\rho_H(t) = (1-t)^{-H+1}$ for $H \in (1,2)$ becomes increasingly difficult with the grid method as $H \to 1$. Moreover, the discretization error is reflected not only at the simulation level, through the autcocorrelation function and moments of the marginal dsitribution, but also when inferring the parameters of the simulated trawl process, as shown in Figure \ref{fig:grid_gmm}. This is particularly important when carrying out simulation-based inference or when comparing two methods to infer the parameters of the trawl process, such as in \cite{RePEc:aah:create:2021-12}. Performing high accuracy or exact simulations ensures that the inference error is entirely due to the inference method rather than the simulation error.
\begin{figure}[ht]
  \centering
  \subfloat[]{\includegraphics[height = 5.75cm,keepaspectratio ]{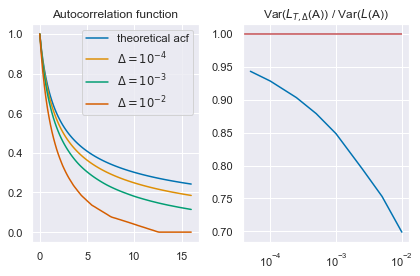}\label{fig:acf_and_var}}
  \hfill
\subfloat[]{\includegraphics[height = 5.75cm,keepaspectratio]{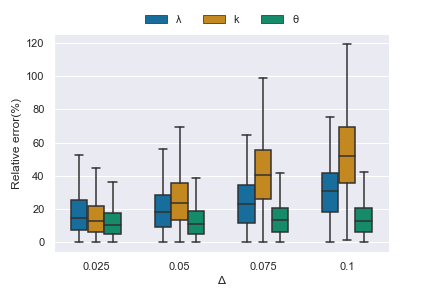}\label{fig:grid_gmm}} 
 \caption{(a) Autocorrelation and variance of the trawl process with trawl function $\phi \colon (-\infty,0] \to \R_{\ge 0}$ given by $\phi(t) = 0.5(1-t)^{-1.5}$, simulated with the grid method for different values of $\Delta.$ For each $\Delta$, we set $T = \Delta^{-0.5}$, which satisfies Theorem \ref{grid_discretisation_thm} and ensures that the discretization error dominates the truncation error. Note that the autocorrelation function and variance are slow to convergence, and that both effects compound on the autocovariance function. By \eqref{eq:mean_L(A)} and \eqref{eq:var_L(A)}, $\Var{\left(L_{T,\Delta}\left(A\right)\right)}/\Var{\left(L\left(A\right)\right)}$ is equal to $\ev\left[L_{T,\Delta}\left(A\right)\right] / \ev\left[L\left(A\right)\right]$ and to the ratio of the area of the cells included in the simulation to the area of the trawl set $A$. (b) Boxplots of relative error for GMM parameter estimates from $1000$ simulations for each of $\Delta = 0.025,\,0.05,\,0.075$ and $0.1$. Here $\tau=0.15,\,k=1000,\,L^{'}\sim \text{Gamma}(k,\theta)$ with $k=2,\,\theta=3$ and $\phi:(-\infty,0] \to \R_{\ge 0}$ given by $\phi(t) = \lambda e^{\lambda t}$ with $\lambda =1.$ We match the first two empirical  moments to infer $k,\,\theta$ and the empirical autocorrelation function at lags $1,\,3,\,5$ to infer $\lambda$. The results appear to be robust to the choice of lags. Note the improvements across all parameters in the relative error as $\Delta \downarrow 0.$}\label{fig:acf_var_gmm}
\end{figure}
The above observations agree with Sections $4$ and $5$  of \cite{nguyen_DG_RG_comparison}, in which the grid method was applied to simulate a Spatio-temporal Ornstein-Uhlenbeck (STOU) process of the form
\begin{equation*}
    Y_t(x) = \int_{A_t(x)} e^{-\lambda(t-s)} L(\mathrm{d}\mathbf{z},\mathrm{d}s).
\end{equation*} 
Further, it was shown in \cite{nguyen_DG_RG_comparison} that tailoring the type of the grid to match the shape of the set $A$ can reduce the effects of the discretization procedure on spatio-temporal correlations and improve convergence. In particular, the authors approximated $Y_t(\mathbf{x})$ by a finite sum on a rectangular grid and on a diamond grid and noticed faster convergence of the diamond grid scheme in some parameter regimes. The error analysis for STOU processes is complicated by the presence of the kernel $e^{-\lambda(t-s)},$ which makes it difficult to distinguish between truncation error, kernel discretization error and L\'evy basis discretization error. All in all, tailoring the grid to each particular trawl shape is time-consuming and the algorithm itself is both memory and computationally expensive. The redundancy can be seen immediately: many of the cells are contained in just one of the sets $A_{\tau},\ldots,A_{k\tau}.$ Thus, instead of simulating all of these cells, we can either discard grid methods and attempt the compound Poisson method, or we could simulate from the slice partition, which can be seen as the optimal grid.  
 
Accompanied by a Gaussian sampler, the compound Poisson scheme can be used to simulate the jump part of the L\'evy basis and is exact when the L\'evy measure $l$ can be normalised to a probability measure. Even if this is not the case, the algorithm has good convergence properties: the approximation $X^\epsilon$ obtained by truncating the jumps at some threshold $\epsilon$ converges in the supremum norm and in $\mathcal{L}^2$ to the trawl process $X$ as $\epsilon \to 0.$ Simulating from the potentially truncated L\'evy measure can be implemented by methods such concave-convex adaptive rejection sampling \cite{ccarj} or MCMC. In general, this may require custom-made samplers or be computationally expensive.

Alternatively, the slice partition method, which can be seen as a natural generalization of the grid discretization, is the only one to result in exact simulation of the trawl process, regardless of the L\'evy measure and trawl set. Further, we can directly trade off the accuracy and speed of the simulation scheme by neglecting slices with small areas, as discussed in Subsection \ref{subsection:slice_partition_algorithm}. For faster, approximate simulation, the choice between the compound Poisson and slice partition schemes depends on the number $k$ of trawls to be simulated, the spacing $\tau$ between the trawls and the difference in the cost of obtaining samples with law $L(A)$ for various values of $\mathrm{Leb}(A)$ versus samples from the potentially truncated L\'evy measure. Generally speaking, if both samples from $L(A)$ and $l$ are available, the slice partition method is suitable for increasing domain simulations where the trawls are sampled at equidistant times, whereas the compound Poisson method is suitable for infill simulations and non-equidistant times.

Finally, note that Algorithms \ref{algo:grid_discretisation} and \ref{algo:slice_partition_trawl_process_convolution} require samples from $L(A)$ for some sets $A$. 
Let $\mu$ be the probability distribution of $L(A)$. In many cases of interest, such as the ones in Subsection \ref{subsection_marginal_distr}, $\mu$ is part of a family of named probability distributions for which efficient samplers are already available. Nevertheless, $\mu$ can also be specified through its cumulant transform $C(\theta,\mu)$.
We show that samples from $\mu$ can be obtained efficiently even in this case, under the mild assumption that the discrete and continuous components of $\mu$ can be separated. Indeed, by the Lebesgue decomposition theorem, $\mu$ can be decomposed into $\mu_{\textrm{ac}} + \mu_{\textrm{sc}} + \mu_{\textrm{d}}$ for some absolutely continuous measure $\mu_{\textrm{ac}}$, some singular continuous measure $\mu_{\textrm{sc}}$ and some discrete measure $\mu_{\textrm{d}} = \sum_{i=1}^\infty p_i \delta_{a_i},$ where $\delta_x$ is the Dirac measure at $x, p_i$ are strictly positive and $a_i$ are non-zero real numbers. To exclude pathological cases, assume that $\mu_{\text{sc}}=0$ and further that given $C(\theta,\mu),$ we can separate $C(\theta,\mu_{\textrm{ac}})$ from $C(\theta,\mu_d) =\sum_{i=1}^\infty e^{i\theta a_i}p_i$. We can sample from $\mu_d$, as it has discrete support. We can also sample from $\mu_{\text{ac}},$ by means of efficiently inverting the Fourier transform of $\mu_{\text{ac}}$, which we discuss next. 

Let $f$ and $F$ be the probability densitity, respectively cumulative distribution functions of $\mu_{\textrm{ac}}.$ The inversion method samples from $\mu_{\textrm{ac}}$ by solving $x = F^{-1}(u)$ where $u$ is generated from the uniform distribution on $(0,1)$. When analytic expressions are not available for $F^{-1},$ a numerical procedure such as Newton-Raphson can be used. This amounts to iterating 
\begin{equation}
    x_{k+1} = x_{k} - \frac{F(x_k)-u}{f(x_k)},
    \label{eq:newton_iterative_procedure}
\end{equation}
from a starting point $x_0$ until a predefined tolerance level is achieved. To account for regions where $F$ is flat, a standard modification of Newton-Raphson can be used, which switches to the bisection method whenever necessary \cite{ridout2009generating}. It remains to approximate the values of $F(x_k)$ and $f(x_k)$ by numerical inversion of the Fourier transform $w \to \int e^{iwx} \mu_{\text{ac}}\left(\mathrm{d}x\right) = e^{C\left(\theta,\mu_{\text{ac}}\right)}$ of $\mu_{\textrm{ac}}.$ There is a rich literature on the topic of Fast Fourier Transform algorithms and quadrature methods for the approximation of probability functions. We mention \cite{hurlimann2013improved}, which provides a comprehensive exposition and error analysis, and \cite{witkovsky2016numerical}, which employs the Gil-Pelaez inversion formula in conjunction with the fast Fourier transform algorithm to draw samples from distributions with known cumulant transform. In the case of distributions supported on the positive or negative real line, such as Gamma and Inverse Gaussian, we can work with the Laplace transform instead of the Fourier transform. \cite{veillette2011technique} demonstrates that the Post-Widder inversion formula can be used to approximate $f$ and $F$ specifically for infinitely divisible distributions. Alternatively, \cite{ridout2009generating} takes a general approach and approximates the Bromwich inversion integral of the Laplace transform of $\mu_{\text{ac}}$ by the trapezium rule and by employing the Euler summation to accelerate convergence. The paper provides an extensive error analysis and an R script, which we adapt for Python and make available at \cite{Leonte_Ambit_Stochastics_2022}. All in all, there are multiple off the shelf algorithms which provide arbitrary accuracy and can aid in the simulation of trawl processes and ambit fields.

%% file: algorithms/grid_discretisation.tex
 \begin{algorithm}
 \caption{Grid discretization}
 \label{algo:grid_discretisation}
  \begin{algorithmic}[1]
      \Require Trawl function $\phi$ and $T<0$ such that $\phi(T)=0;$ number of steps $N_t$,$N_x$ on the time and space axes; sampler $S(\textrm{area},m,n)$ which returns an $m\times n$ array of iid samples with the same law as $L(A),$ where $\mathrm{Leb}(A) = \textrm{area};$ number of trawls to be simulated $k$ and distance $\tau$ between them.
    \Ensure Vector $X$ containing the simulated values of the trawl process at times $\tau,\ldots,k\tau.$
      \Function{ComputeIndicator}{$\phi,\,T,\,N,\,N_x,\,\Delta_t,\,\Delta_x$} \Comment{Helper function}
        \State $I \gets \text{zeros}(N_x,N)$ 
        \For{$i = 1,\ldots,N$}
            \For{$j=1,\ldots,N_x$}
                \IfThen {$i\Delta_x \le \phi(T + (j-1)\Delta_t)$}{$I_{ij} \gets 1$} \Comment{Check if $g^l_{ij}$ is contained in $A_{l\tau}$ (see Remark \ref{remark_cell_and_lebesgue_decomp})}

            \EndFor
        \EndFor
        \State \Return I
      \EndFunction
      \Function{main}{$\phi,\,T,\,N_t,\,N_x,\,S,\,k,\,\tau$}
      \State $T \gets \floor{T/\Delta_t} \Delta_t$ \Comment{Ensures all cells have the same area and excludes boundary effects}
      \State $\Delta_t \gets \tau / N_t,\, \Delta_x \gets \phi(0) / N_x,\, N \gets - T/ \Delta_t$
      \State $X \gets \textrm{zeros}(k)$
      \State $I \gets \Call{ComputeIndicator}{\phi,\,T,\,N,\,N_x,\,\Delta_t,\,\Delta_x }$
      \State $Y_1 \sim S(\Delta_t\Delta_x,\,N,\,N_x) $ \Comment{Sample $L(c)$ for $c \in G_1$}
      \State $X[1] \gets Y_1 \odot I $
      \For{$2 = 1,\ldots,k$}
        \State $Y_l \gets \textrm{zeros}(N,N_x)$ \Comment{Update step (see Figure \ref{fig:grid_update})}
        \State $Y_l[:,1:N-N_t] \gets Y_{l-1}[:,N_t+1:N]$ \Comment{Keep $L(c)$ for $c \in G_{l+1}\cap G_l$}
        \State $Y_l[:,N-N_t+1:N] \sim S(\Delta_t\Delta_x,\,N_x,\,N_t)$ \Comment{Sample $L(c)$ for $c \in G_{l+1} \backslash G_l$}
        \State $X[l] \gets Y_l \odot I$
      \EndFor
        \Return X
      \EndFunction
    \end{algorithmic}
\end{algorithm}

%% file: pictures/tikz_files/grid_approx_and_update.tex
\begin{figure}
    \centering
    \begin{subfigure}[b]{0.47\textwidth}
        \resizebox{\linewidth}{!}{
        \includegraphics{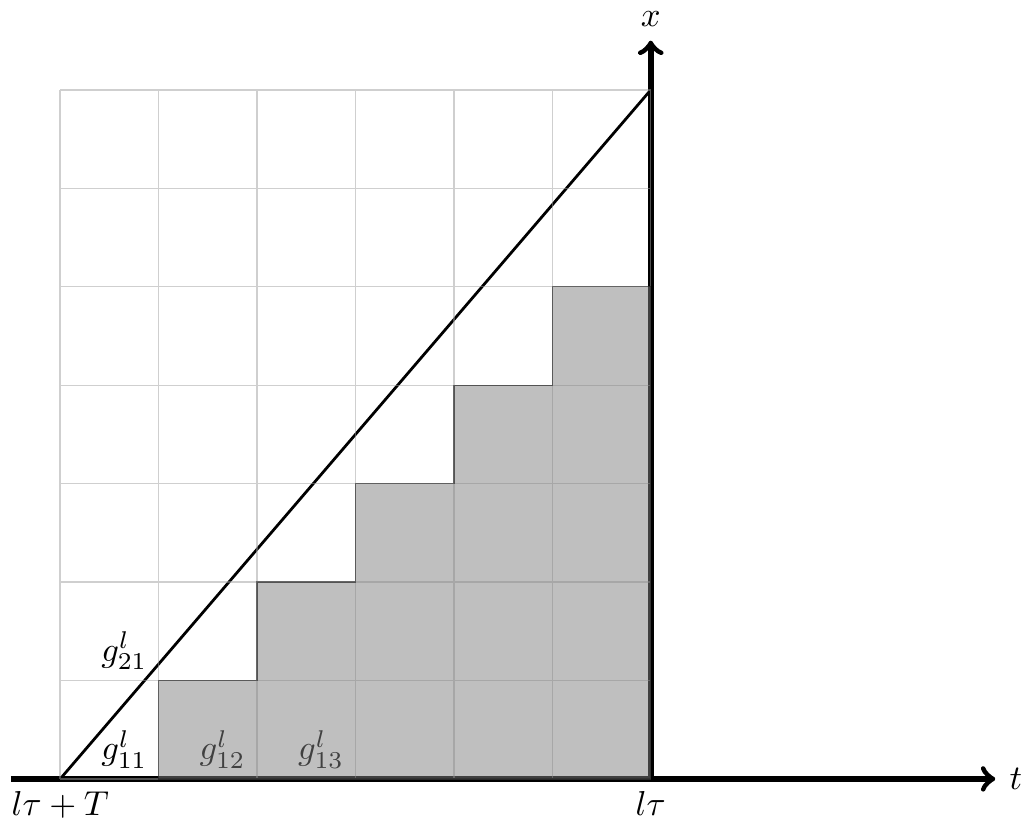}
        }
        \caption{Illustration of the grid discretization algorithm: the L\'evy basis evaluated of the ambit set $A_{l\tau}$ is approximated by $\sum{L(c)},$ where the sum is taken over cells $c$ which are fully contained in $A_{l\tau}.$ These cells are shaded in gray. \\
        }
        \label{fig:included_cells}
    \end{subfigure}
    \hspace{0.5cm}
    \begin{subfigure}[b]{0.47\textwidth}
        \resizebox{\linewidth}{!}{
        \includegraphics{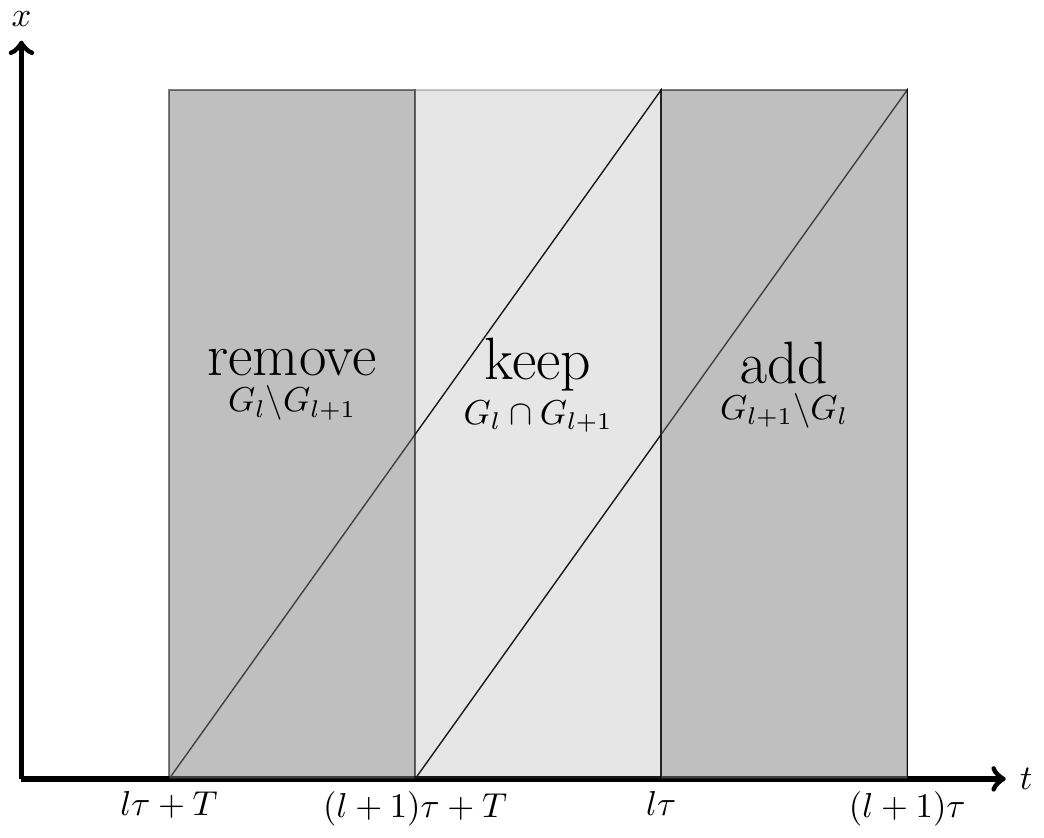}
        }

        \caption{Illustration of the grid update procedure at step $l$: we remove from memory the values $L(c)$ corresponding to cells whose 'upper-left' corners have time coordinates less than $(l+1) \tau+T$ and simulate $L(c)$ over cells whose 'upper-left' corners have time coordinates in the interval $[l\tau,(l+1)\tau].$}        \label{fig:grid_update}
    \end{subfigure}
    \caption{}
\end{figure}

%% file: pictures/slice/Y_tilde_compact.tex
\begin{equation}
Y = 
\begin{pmatrix}
0  & \ldots  & 0 &          L\left(S_{11}\right)  & \ldots &  L\left(S_{1k}\right) \\
 0  & \ldots  & 0 &        L\left(S_{21}\right)  &  \ldots &  L\left(S_{2k}\right)\\
         \vdots        &    & \vdots & \vdots & & \vdots  \\
\undermat{I-1}{0  & \ldots  & 0 } &       L\left(S_{I1}\right) &  \ldots &  L\left(S_{Ik}\right)\\
\end{pmatrix}, \  F = \begin{pmatrix}
    0 & 0 & \dots & 0 & 1 \\
    0 & 0 & \dots & 1 &1 \\
    \vdots & \vdots & \udots & \vdots & \vdots \\
    0 & 1 & \dots & 1 &1 \\
    1 & 1 & \dots & 1 & 1
  \end{pmatrix}.\label{eq:f_matrix}
\end{equation}

%% file: algorithms/slice_partition_algo.tex
 \begin{algorithm}
        \caption{Slice partition for bounded, monotonic trawls}\label{algo:slice_partition_trawl_process_convolution}
  \begin{algorithmic}[1]
      \Require Sampler $S(\textrm{area})$ which returns independent samples with the same law as $L(A),$ where $\mathrm{Leb}(A) = \textrm{area};$ number of trawls to be simulated $k$ and distance $\tau$ between them; $I = \ceil{-T/\tau}$.
    \Ensure Vector containing the simulated values of the trawl process at times $\tau,\ldots,k\tau.$
    \Function{main}{$S,\,k,\,\tau,\,I$}
        \State Compute the areas $s_{ij}$ from  \eqref{trawl_1_areas},\eqref{trawl_2_areas} \Comment{Requires $I$ integrations and $I$ differences}
        \State $F \gets \text{tril}(I)$ \Comment{$I \times I$ lower triangular matrix of ones as in \eqref{eq:f_matrix}}
        \State $Y \gets \text{zeros}(I,I+k-1)$
        \For{$i = 1,\ldots,I$}
         \For{$j=1,\ldots,k$}
            \State $Y[i,I-1+j] \gets S\left(s_{ij}\right)$
                \EndFor    
            \EndFor
      \Return $Y * F$ \Comment{Convolution step}
      \EndFunction
    \end{algorithmic}
\end{algorithm}

%% file: pictures/slice/Y_tilde_non_compact.tex
\begin{equation*}
Y = 
\begin{pmatrix}
0  & \ldots  & 0 &          L\left(S_{11}\right)  & \ldots & L\left(S_{1,k-1}\right) & L\left(S_{1k}\right) \\
 0  & \ldots  & 0 & L\left(S_{21}\right)  &  \ldots &  L\left(S_{2,k-1}\right)  & 0\\
         \vdots        &    & \vdots & \vdots & \udots & \vdots   & \vdots \\
\undermat{k-1}{0  & \ldots  & 0 } &      L\left(S_{k1}\right) &  \ldots &  0 & 0\\
\end{pmatrix}, \   F = \begin{pmatrix}
    0 & 0 & \dots & 0 & 1 \\
    0 & 0 & \dots & 1 &1 \\
    \vdots & \vdots & \udots & \vdots & \vdots \\
    0 & 1 & \dots & 1 &1 \\
    1 & 1 & \dots & 1 & 1
  \end{pmatrix}, 
\end{equation*}

%% file: Extensions_to_kernel_weighted_trawls.tex
\section{Extensions to volatility modulated, kernel-weighted trawl processes}\label{section:extensions_to_VMKWTP}
Trawl processes are stationary, infinitely divisible and ergodic processes which can describe a wide range
of possible serial correlation patterns in data. Nevertheless, many systems of interest are inherently
non-stationary; examples include precipitation data from \cite{slater2020nonstationary} and financial time series. To model such behaviour,
steps are usually taken to transform the initial process into a stationary one, to which standard methods
can be applied. Furthermore, recent empirical work in areas such as environmental sciences in \cite{huang2011class} and energy pricing in \cite{benth_electricity,9966b470f5bd4eb5923297a0a31afb74} shows the presence of volatility clusters,
and thus of stochastic volatility, which cannot be replicated by traditional models. In the following, we
show that both non-stationarity and stochastic volatility can easily be incorporated into the trawl process
framework and the same algorithms can be used for efficient simulation, despite the more complicated
structure. We present the non-stationary case first and then the volatility modulated one. Finally, we
extend the aforementioned methodology to the spatio-temporal case in Section \ref{section:extensions_to_ambit_field_simulation}.

Since trawl processes are given by the L\'evy basis evaluated over the trawl sets, and depending on the
kind of desired non-stationary behaviour, we can change either the trawl sets or the L\'evy basis. At the level of the trawl sets, we can use different trawl functions $\phi_t \colon (-\infty,0] \to \R_{\ge 0}$ to define $A_t = \{(s, x) \colon \ 0 <
x < \phi_t(s-t)\}$, which changes both the correlation structure and marginal distributions. At the level of the L\'evy basis, we can drop the homogeneity assumption of Definition \ref{def:levy_basis}, which allows the distribution of the jumps of $L_j$ to depend on the point $\mathbf{z}=(t,x)$ where they appear, the jumps to have a nonuniform intensity measure, and even for the interaction between the jump values and the jump intensity. Similarly, the distribution of $L_g(A)$ is allowed to depend on the points in $A,$ not just on $\mathrm{Leb}(A).$ The following result from \citep[Proposition 2.1]{rajput1989spectral} reflects the extra flexibility of inhomogeneous L\'evy basis over homogeneous ones.
\begin{lemma}[Cumulant of L\'evy bases]
Let $L$ be a L\'evy basis on $S \subset \R^d, \theta$ a real number and $A$ in $\mathcal{B}_b(S).$ Then
\begin{equation}
   C\left(\theta,L(A)\right) = \int_A \left(i \theta \xi(\mathbf{z}) - \frac{1}{2} \theta^2 a(\mathbf{z}) + \int_{\R}\left(e^{i \theta y}-1-i \theta y\mathbf {1}_{[-1,1]}(y)\right) l(\mathrm{d}y,\mathbf{z}) \right) c(\mathrm{d}\mathbf{z}),\label{eq:cumulant transform}
\end{equation}
\end{lemma}
where $\xi \colon S \to \R$, $a \colon S \to \R_{\ge0}$, $l(\cdot,\mathbf{z})$ is a L\'evy measure on $\R$ for each $\mathbf{z}$ in $S$ and $c$ is a measure on $S,$ called the control measure, such that the above integral is well defined. Similarly to Definition \ref{def:levy_seed}, functions $\xi$ and $a$ give the local drift and variance of the Gaussian component at $\mathbf{z}.$

To each $\mathbf{z} = (t,\mathbf{x})$ in $S,$ we can associate an infinitely divisible random variable $L^{'}(z)$ with
\begin{equation*}
    C\left(\theta,L^{'}\left(\mathbf{z}\right)\right) =  i \theta \xi(\mathbf{z}) - \frac{1}{2} \theta^2 a(\mathbf{z}) + \int_{\R}\left(e^{i \theta y}-1-i \theta y\mathbf {1}_{[-1,1]}(y)\right) l(\mathrm{d}y,\mathbf{z}), 
\end{equation*}
which we call the L\'evy seed at $\mathbf{z}.$ Then \begin{equation*}
C\left(\theta,L(A)\right) = \int_A C\left(\theta,L^{'}\left(z\right)\right) c\left(\mathrm{d}\mathbf{z}\right),
\end{equation*}
and the distribution of $L^{'}\left(\mathbf{z}\right)$ can be seen as the distribution of the infinitesimal $L\left(\mathrm{d}\mathbf{z}\right).$ As opposed to the homogeneous case in Definition \ref{def:levy_seed}, this is now a function of $\mathbf{z}.$ Further, to exclude pathological cases, assume that $c$ has no discrete or singular continuous part. In this case, without loss of generality, we can incorporate the Radon–Nikodym derivative $\frac{\mathrm{d}c}{\mathrm{d}Leb}$ into $\xi(\mathbf{z})$, $a(\mathbf{z})$, $l(\mathrm{d}y,\mathbf{z})$ and set $c = \mathrm{Leb}.$ Then the intensity function of the jumps at $\mathbf{z}$ is given by the total mass $l(\R,\mathbf{z})$.

Despite the great flexibility, we are not aware of settings, apart from theoretical study, where inhomogeneous L\'evy bases have been used. We propose the kernel-weighted trawl processes, which encompasses both the changes in the trawl sets and L\'evy basis by introducing a kernel with respect to a homogeneous L\'evy basis. This results in a more general class of trawl processes with compact notation for which the slice partition and compound Poisson simulation algorithms can be applied with few or no modifications.
\begin{definition}\label{definition:kernel_weighted_trawl_processes}
Let $S \subset \R^d,\ K_t \colon (S,\mathcal{B}_{\mathrm{Leb}}(S)) \to (\R,\mathcal{B}\left(\R\right))$ be a family of measurable mappings and $A \in  \mathcal{B}_{\mathrm{Leb}}(S)$. The kernel-weighted trawl process over the collection of trawl sets $A_t = A + (t,0)$ is given by
\begin{equation*} \label{eq:kernel_weighted_def}
    X_t = \int_{A_t} K_t(\bar{t},\bar{\mathbf{x}}) L(\mathrm{d}\bar{t} ,\mathrm{d} \bar{\mathbf{x}}),
\end{equation*}
\end{definition}
under mild regularity conditions of the kernel $K$ and where the integration is understood in the $\mathcal{L}^0$ framework of \citep[Theorem 2.7]{rajput1989spectral}. We mention that it is enough for the $\mathcal{L}^2\text{-norm}$ of the kernel over the trawl sets $t \to \int_{A_t} K_t^2\left(\bar{t},\bar{\mathbf{x}}\right) \mathrm{d}\bar{t}\mathrm{d}\bar{\mathbf{x}}$ to be bounded on compacts $[0,T]$ for the existence conditions to be satisfied. Similarly to the case of trawl processes, the cumulant transform of $X_t$ is given by \citep[Proposition 2.6]{rajput1989spectral}
\begin{equation}
    C\left(\theta,X_t\right) =C\left(\theta,\int_{A_t} K_t(\bar{t},\bar{\mathbf{x}})  L\left(\mathrm{d}\bar{t},\mathrm{d}\bar{\mathbf{x}}\right)\right)  = \int_{A_t} C\left(\theta K_t(\bar{t},\bar{\mathbf{x}}),L^{'}\right) \mathrm{d}\bar{t}\mathrm{d}\bar{\mathbf{x}},\label{eq:general_cumulant}
\end{equation}
and if $\Var\left(L^{'}\right)$ is finite, the second order structure is given by
\begin{align*}
\ev\left[X_t\right] &= \ev\left[L^{'}\right]  \int_{A_t} K_t(\bar{t},\bar{\mathbf{x}})\mathrm{d}\bar{t}\mathrm{d}\bar{\mathbf{x}}, \\
\Cov\left(X_t,X_s\right) &= \Var\left(L^{'}\right) \int_{A_t \cap A_s} K_t(\bar{t},\bar{\mathbf{x}}) K_s(\bar{t},\bar{\mathbf{x}})\mathrm{d}\bar{t}\mathrm{d}\bar{\mathbf{x}}.
\end{align*}
We discuss the simulation of the kernel-weighted trawl process in increasing order of complexity. In Subsection \ref{subsection:det_volatility}, we consider processes of the form $    X_t = \int_{A_t} K(\bar{t},\bar{\mathbf{x}}) L(\mathrm{d}\bar{t} ,\mathrm{d} \bar{\mathbf{x}})$. The kernel $K(\bar{t},\bar{\mathbf{x}})$ can be thought of as the nonstationary component of the trawl process, or equally as a deterministic volatility. In Subsection \ref{subsection:t-dependent-kernels}, we extend to time dependent kernels  $K_t\left(\bar{t},\bar{\mathbf{x}}\right)$, which allow for a more complicated joint distribution and, in particular, for the autocorrelation function to take both positive and negative values. Further, we discuss how the joint structure can be controlled solely through the kernel, by using a simple shape for the ambit set, such as a rectangle. In this setting, we recover the Brownian and Levy semistationary processes. Finally, in Subsection \ref{subsection:volatility_modulated_trawls}, we discuss modulation of the trawl process by a stochastic volatility. As in Section \ref{section:simulation}, we only discuss the simulation in the case $S= \R^2;$ generalizations to $S \subset \R^d$ for $d>2$ are straightforward.
\subsection{The non-stationary and deterministic volatility cases}\label{subsection:det_volatility}
We can directly generalize the slice partition method to simulate the kernel-weighted trawl process at times $\tau,\ldots,k\tau.$ We decompose the sets $A_\tau,\ldots,A_{k\tau}$ into a collection  $\mathcal{S}$ of disjoint slices $S$, sample $\int_S K(\bar{t},\bar{x})L(\mathrm{d}\bar{t} ,\mathrm{d} \bar{x})$ for all $S \in \mathcal{S}$ and set
\begin{equation*}
    X_{l\tau} = \sum_{S \subset A_{l\tau}}  \int_{S }K(\bar{t},\bar{x}) L(\mathrm{d}\bar{t} ,\mathrm{d} \bar{x}).
\end{equation*}
As discussed in Subsection \ref{subsection:summary_and_sampler}, the above sampling may be done analytically or may involve numerically inverting the cumulant transform by a Fast Fourier Algorithm or the Laplace transform by approximating the inversion integral over a contour in the complex plane, such as the Bromwich contour. In turn, this requires multiple evaluations of the integrand, which is itself given as an integral in \eqref{eq:general_cumulant}. If the integral in \eqref{eq:general_cumulant} is difficult to compute, we can separate the kernel-weighted trawl process into $X_t = \int_{A_t} K(\bar{t},\bar{x}) L_g(\mathrm{d}\bar{t} ,\mathrm{d} \bar{x}) + \int_{A_t }K(\bar{t},\bar{x}) L_j(\mathrm{d}\bar{t} ,\mathrm{d} \bar{x})$ and simulate the Gaussian part via the slice partition and the jump part via the compound Poisson method. For the Gaussian part, $\int_{S }K(\bar{t},\bar{x}) L_g(\mathrm{d}\bar{t} ,\mathrm{d} \bar{x}) \sim \mathcal{N}\left(\mu \int_{S} K(\bar{t},\bar{x})\mathrm{d}\bar{t} \mathrm{d} \bar{x}, \sigma^2 \int_{S} K^2(\bar{t},\bar{x})\mathrm{d}\bar{t} \mathrm{d} \bar{x} \right)$, where $L_g^{'} \sim \mathcal{N}\left(\mu,\sigma^2\right)$ and which requires only one integral evaluation per slice. For the jump part, note that $L_j$ is a discrete measure supported on at most countably many points $(t_i,x_i)$, with jumps $y_i$ distributed according to the L\'evy measure. In particular, we have $L_j = \sum_i y_i \delta_{(t_i,x_i)}.$ Then $ \int_{A_t}K(\bar{t},\bar{x}) L_j(\mathrm{d}\bar{t} ,\mathrm{d} \bar{x}) = \sum_i y_i K(\bar{t_i},\bar{x}_i)$. If the L\'evy measure is infinite, then we truncate at some small jump threshold $\epsilon$, as in Subsection \ref{subsection: cpp}, and $X_t$ is approximated by $\sum_{i : y_i > \epsilon} y_i K(\bar{t_i},\bar{x}_i).$ The uniform convergence of Theorem \ref{theorem:uniform_convergence} is still valid for kernels $K$ for which $t \to \int_{A_t} K^2\left(\bar{t},\bar{x}\right) \mathrm{d}\bar{t}\mathrm{d}\bar{x}$ is bounded on compacts $[0,T]$, as sketched in Remark \ref{remark:uniform_convergence_for_KW}.

Figure \ref{fig:kwt_no_t} displays all the possibilities. Figure \ref{fig:kwt_no_t_poisson} has $L^{'}\sim \text{Poisson}(5)$, i.e.~a finite L\'evy measure for which the compound Poisson approach can be used for exact simulation and an integer-valued kernel $K\left(\bar{t},\bar{x}\right) = \floor{2 \left(1+\bar{x}\right)(\bar{t}- \floor{\bar{t}})}$; note the sharp peaks induced by $\left(1+\bar{x}\right),$ which assigns larger values to points with larger $\bar{x}$ coordinate, i.e. points which 'leave' the trawl sets faster and further note the periodic trend induced by $(t- \floor{t})$. Figure \ref{fig:kwt_no_t_cauchy} has $L^{'} \sim \text{Cauchy}(1)$ and $K\left(\bar{t},\bar{x}\right) = \sin{\bar{t}}$; then $\int_S K\left(\bar{t}\right) \sim \text{Cauchy}\left(\int_S K\left(\bar{t},\bar{x}\right) \mathrm{d}\bar{t}\mathrm{d}\bar{x}\right) $ and the slice partition can be applied without modification. Figure \ref{fig:kwt_no_t_gaussian_gamma} has $L^{'} \sim \mathcal{N}(2,4) + \text{Gamma}(2,2)$ and $K\left(\bar{t},\bar{x}\right) = 1 + 0.1 \bar{t}.$ The numerical evaluation of the cumulant transform from \eqref{eq:general_cumulant} for multiple values of $\theta$ is expensive, hence we use the slice partition method for the Gaussian part and the compound Poisson method for the jump part; note the increasing trend.
 \begin{figure}[h]
        \centering
        \begin{subfigure}[t]{0.32\textwidth}
         \includegraphics[width=\textwidth]{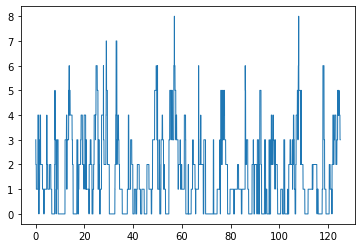}
        \caption{$L^{'}\sim \text{Poisson}(5)$, $K\left(\bar{t},\bar{x}\right) = \floor{2 \left(1+\bar{x}\right)(\bar{t}- \floor{\bar{t}})}$.}
        \label{fig:kwt_no_t_poisson}
        \end{subfigure}
        \hfill
        \begin{subfigure}[t]{0.32\textwidth}
           \includegraphics[width=\textwidth]{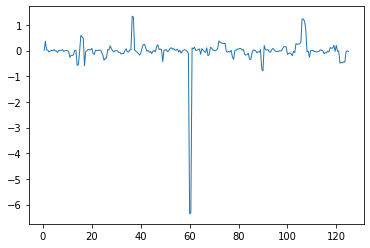}
            \caption{$L^{'} \sim \text{Cauchy}(1)$, $K\left(\bar{t},\bar{x}\right) = 0.1\sin{\bar{t}}$.}
            \label{fig:kwt_no_t_cauchy}
        \end{subfigure}
        \hfill
        \begin{subfigure}[t]{0.32\textwidth}
           \includegraphics[width=\textwidth]{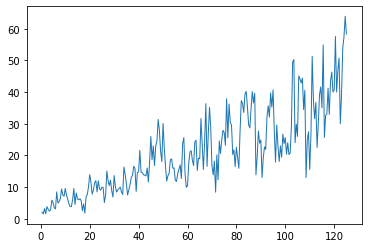}
            \caption{$L^{'} \sim \mathcal{N}(2,4) + \text{Gamma}(2,2)$, $K\left(\bar{t}\right) = 1 + 0.1 \bar{t}.$}
            \label{fig:kwt_no_t_gaussian_gamma}
        \end{subfigure}
        \caption{Simulation of three kernel-weighted trawl processes between $t=0$ and $t=120$ for a triangular trawl set $A = \{(s,x) : -2 < s < 0, 0 < x < 1 + s/2\} $ in a) and exponential trawl set $A = \{(s,x) : s <0, 0 < x < e^{s}\}$ in b) and c); a) is simulated exactly, b) and c) are simulated at equidistant times $\tau,\ldots,250\tau,$ where $\tau=0.5$. Whereas b) just requires one integration per slice to determine the distribution of $\int_S K\left(\bar{t},\bar{x}\right) \mathrm{d}\bar{t}\mathrm{d}\bar{x}$, c) also requires drawing samples from the truncated L\'evy measure $l^\epsilon$ with density $\frac{\mathrm{d}l^{\epsilon}}{\mathrm{d}Leb}\left(y\right) = 2y^{-1} e^{-3y} 1_{y > \epsilon}$, where we choose $\epsilon = 5 \cdot 10^{-3}.$ We draw samples from $l^\epsilon$ by rejection sampling with a convex envelope of the log density, as inspired by \cite{ccarj}.}
     \label{fig:kwt_no_t}
\end{figure}
\subsection{Time dependent kernels}\label{subsection:t-dependent-kernels}
The extra dependence of $K_t(\bar{t},\bar{x})$ on $t$ allows for a different kernel for each trawl set $A_t$ and results in a more general joint distribution.
The increased flexibility is matched by an increase in computational cost: the slice partition method can only be used for the Gaussian part, as dependent random variable sampling is difficult outside the Gaussian framework. Further, the simulated slices are not reusable, in the sense that we need to sample up to $k$ values $\int_{S} K_\tau(\bar{t},\bar{x}) L_g(\mathrm{d}\bar{t} ,\mathrm{d} \bar{x}),\ldots,\int_{S} K_{k\tau}(\bar{t},\bar{x}) L_g(\mathrm{d}\bar{t} ,\mathrm{d} \bar{x})$ for each slice in the partition induced by $X_\tau,\ldots,X_{k\tau}$. 
Similarly, the approximation of $X_t \approx \sum_{i : y_i > \epsilon} y_i K_t(\bar{t_i},\bar{x}_i)$ now requires up to $k$ evaluations of the kernels at each jump location $(t_i,x_i)$, the number of evaluations corresponding to the number of trawl sets $A_\tau,\ldots,A_{k\tau}$ which contain $(t_i,x_i)$. Unlike the general algorithms from Section \ref{section:simulation}, the efficient simulation of kernel-weighted trawl processes depends heavily on the special structure to be exploited in each setting. In particular cases, we are still able to pull back to these off the shelf methods. We study one such example, inspired by \cite{MR3163237}. Consider the generalized Ornstein–Uhlenbeck(OU) process  $X_t = \int_{A_t} e^{\lambda(\bar{t}-t)} L(\mathrm{d}\bar{t},\mathrm{d}\bar{x})$. By the multiplicative property of the exponential function, we can break $e^{\lambda(\bar{t}-t)}$ into $e^{\lambda\bar{t}} e^{-\lambda t},$
simulate $\int_{A} e^{\lambda\bar{t}} L(\mathrm{d}\bar{t},\mathrm{d}\bar{x}), \ldots, \int_{A_{k\tau}} e^{\lambda\bar{t}} L(\mathrm{d}\bar{t},\mathrm{d}\bar{x})$ as discussed in Subsection \ref{subsection:det_volatility} and then multiply the above values by $e^{- \lambda \tau}, \ldots,e^{- \lambda  k \tau} $. The same procedure is applicable if $K$ is given by $\sin{(\bar{t}-t)}$, $\cos{(\bar{t}-t)}$, a linear combination of sines and cosines or if $K$ is well approximated by such a linear combination.

More generally, consider $K_t(\bar{t},\bar{x}) = g(\bar{t}-t)$ for some square integrable $g$ and $X_t = \int_{A_t} g(\bar{t}-t) L(\mathrm{d}\bar{t},\mathrm{d}\bar{x})$ for some L\'evy basis $L$ with finite $\Var{(L^{'})}$. Although this type of kernel-weighted trawl process is stationary, it is strictly more general than a trawl process, and it can exhibit negative correlations. The Fourier expansion methodology from Section $3$ of \cite{MR3163237} can be adapted to show that for a slice $S \subset \{(t,x) : (l-1)\tau \le  t \le l \tau\}$, we have
\begin{align*}
    \int_{S}  g(\bar{t}-t) L(\mathrm{d}\bar{t}\mathrm{d}\bar{x}) &\approx e^{-\lambda l\tau}   \int_S e^{\lambda \bar{t}}  \left(a_0 + \sum_{i=1}^N a_i \cos{\left(n \pi (\bar{t}-l\tau)/\tau\right)}\right) L\left(\mathrm{d}\bar{t},\mathrm{d}\bar{x}\right) \\
    &= e^{-\lambda l \tau} \left( a_0 \int_S e^{\lambda \bar{t}} L\left(\mathrm{d}\bar{t},\mathrm{d}\bar{x}\right)  +  \Re \left(\sum_{i=1}^N a_i e^{- i n l\pi} \int_S e^{ i n \pi \bar{t} / \tau} L\left(\mathrm{d}\bar{t},\mathrm{d}\bar{x}\right)  \right) \right),
\end{align*}
where the approximation is understood in $\mathcal{L}^2,\,N$ is the number of terms in the approximation, $a_0,\ldots,a_N$ are constants and $\lambda$ is a parameter to be calibrated. In the above formula, we require $N$ evaluations per slice, as compared to up to $k$ evaluations. This approximation removes the $t$ dependency of the kernel and works well as long as $N$ is smaller than $k$. The trade off is that the convergence is just in $\mathcal{L}^2$, and not uniformly on compacts. An observation on the type of kernels to be used is in order. Since the autocorrelation function can be modelled through the shape of the trawl set $A,$ the kernel can in principle be chosen from a family of straightforward functions, with which we can work easily, with the aim of inducing drift, seasonal behaviour or deterministic volatility. Consequently, for most practical purposes, it is enough to consider kernels to which the above simplifications apply. 

Finally, note that the formulation in Definition \ref{definition:kernel_weighted_trawl_processes} and the $\mathcal{L}^0$ integration framework of \cite{rajput1989spectral} do not require $A$ to be bounded. Indeed, provided the kernel integrability conditions are satisfied, we can choose a trawl set with a simple geometry, such as an unbounded rectangle $A =\{(s,x) : s <0, 0 < x < 1\}$ and control the joint distribution solely through the kernel. If $K$ is chosen to depend only on $t$ and $\bar{t}$, we recover the L\'evy semistationary processes $X_t = \int_{-\infty}^t K_t(\bar{t})H(\mathrm{d}\bar{t})$, where $H$ is a two-sided L\'evy process with $H_1 \stackrel{d}{=} L^{'}$ and which are extensively studied in \citep[Chapters 1,2 and 10]{ambit_book}.

\subsection{Volatility modulated trawls}
\label{subsection:volatility_modulated_trawls}
A stochastic volatility $\sigma$ can easily be added to the trawl process framework.
\begin{definition}[Kernel-weighted, volatility modulated trawl processes]
\label{definition:kernel_weighted,volatility_modulated_trawl_processes}
Let $S \subset \R^d,\ K_t \colon (S,\mathcal{B}_{\mathrm{Leb}}(S)) \to (\R,\mathcal{B}\left(\R\right))$ be a family of measurable mappings and $A \in \mathcal{B}_{\mathrm{Leb}}(S)$. Let $\sigma$ be a stochastic process on the same probability space as $L$. The kernel-weighted, volatility modulated trawl process over the collection of trawl sets $A_t = A + (t,0)$ is given by
\begin{equation*} \label{eq:kernel_weighted__vol_modulated_def}
    X_t = \int_{A_t} K_t(\bar{t},\bar{\mathbf{x}}) \sigma(\bar{t}) L(\mathrm{d}\bar{t} ,\mathrm{d} \bar{\mathbf{x}}).
\end{equation*}
\end{definition}
In general, if $\sigma$ and $L$ are dependent, the integration is understood in the sense of \cite{walsh} and \cite{bichteler1983random}. We restrict our attention to the case in which $\sigma$ and $L$ are independent. Then the integration can be defined conditionally on $\sigma,$ using the same $\mathcal{L}^0\text{-framework}$ of \citep[Theorem 2.7]{rajput1989spectral}, as for Definition \ref{definition:kernel_weighted_trawl_processes}. 
The second order structure is given by
\begin{align}
\ev\left[X_t | \mathcal{F}_\sigma^t \right] &= \ev\left[L^{'}\right]  \int_{A_t} K_t(\bar{t},\bar{\mathbf{x}}) \sigma(\bar{t}) \mathrm{d}\bar{t}\mathrm{d}\bar{\mathbf{x}},\label{eq:conditional_mean} \\
\Cov\left(X_t,X_s\right | \mathcal{F}_\sigma^t \vee \mathcal{F}_\sigma^s) &= \Var\left(L^{'}\right) \int_{A_t \cap A_s} K_t(\bar{t},\bar{\mathbf{x}}) K_s(\bar{t},\bar{\mathbf{x}}) \sigma^2(\bar{t}) \mathrm{d}\bar{t}\mathrm{d}\bar{\mathbf{x}},\label{eq:conditional_var}\\
   C\left(\theta,X_t\right | \mathcal{F}_\sigma^t) &=C\left(\theta,\int_{A_t} K_t(\bar{t},\bar{\mathbf{x}})  \sigma{(\bar{t}}) L\left(\mathrm{d}\bar{t},\mathrm{d}\bar{\mathbf{x}}\right)\right)  = \int_{A_t} C\left(\theta K_t(\bar{t},\bar{\mathbf{x}}) \sigma{(\bar{t})},L^{'}\right) \mathrm{d}\bar{t}\mathrm{d}\bar{\mathbf{x}}, \label{eq:conditional_cumulant}
\end{align}
where $\mathcal{F}_\sigma^t$ is the $\sigma$ algebra generated by $\{\sigma(\bar{t}) : (\bar{t},\bar{\mathbf{x}}) \in A_t \text{ for some }\bar{\mathbf{x}}\}$. The unconditional structure follows by integrating taking the expectation over $\sigma$. An extensive presentation can be found in \citep[Chapter 5.3.2.1]{ambit_book}.  We restrict our attention to the case $S = \R^2$.

The observations from Subsection \ref{subsection:t-dependent-kernels} still apply: the slice partition can be used for the Gaussian part and the compound Poisson method for the jump part; in some cases, the $t$ dependence can be removed by means of a Fourier approximation. As explained before, the difficulty in sampling by numerically inverting the cumulant is that every step in the iterative procedure of \eqref{eq:newton_iterative_procedure} requires evaluations of the cumulant from \eqref{eq:conditional_cumulant} for multiple values of $\theta$. In turn, each of these evaluations requires the values of $\sigma$ and $K$ for multiple arguments. In this situation, it is usually more expensive to simulate the volatility $\sigma$ than to evaluate the kernel $K$, hence inverting the cumulant may not be practical. The difference between the general case and that of a Gaussian L\'evy basis $L_g$ is that conditionally on $\sigma$, the distribution of $\int_S K(\bar{t},\bar{x}) \sigma(\bar{t}) L(\mathrm{d}\bar{t},\mathrm{d}\bar{x})$ is fully specified by the two integrals $\int_S K(\bar{t},\bar{x}) \sigma(\bar{t}) \mathrm{d}\bar{t}\mathrm{d}\bar{x}$ and $\int_S K^2(\bar{t},\bar{x}) \sigma^2(\bar{t}) \mathrm{d}\bar{t}\mathrm{d}\bar{x}$, and sampling does not require other integral evaluations. The underlying property is that of closure under linear combinations and is satisfied by the family of L\'evy stable distributions from Example \ref{example:stable}. In particular, if $L^{'}\sim \textrm{Stable}(\alpha,\,\beta,\,c,\,\mu)$ with $\alpha \neq 1$, then $\int_S K(\bar{t},\bar{x}) \sigma(\bar{t}) L(\mathrm{d}\bar{t},\mathrm{d}\bar{x}) | \mathcal{F}_\sigma^t \sim \mathrm{Stable}(\alpha,\tilde{\beta},\tilde{c},\tilde{\mu})$, where $\tilde{c} = c \left(\int_S \abs{K(\bar{t},\bar{x}) \sigma(\bar{t})}^\alpha \mathrm{d}\bar{t}\mathrm{d}\bar{x}\right)^{1/\alpha}$, $\tilde{\beta} = \frac{\beta}{\tilde{c}} \left(\int_S \abs{K(\bar{t},\bar{x}) \sigma(\bar{t})}^\alpha \sign{\left(K(\bar{t},\bar{x}) \sigma(\bar{t})\right)} \mathrm{d}\bar{t}\mathrm{d}\bar{x}\right)^{1/\alpha}$ and $\tilde{\mu} = \int_S K(\bar{t},\bar{x}) \sigma(\bar{t}) \mathrm{d}\bar{t}\mathrm{d}\bar{x}$, provided the integrals are finite. In the definition of $\tilde{\beta}$, we write $x^\alpha $ for $\sign{(x)} \abs{x}^{1/\alpha}$. A similar, simpler formula holds for $\alpha=1$, which corresponds to the the Cauchy distribution translated by a location parameter. Consequently, as longs as $L^{'}$ has a stable distribution and the terms $\tilde{c}$, $\tilde{\beta}$, $\tilde{\mu}$ can be approximated well, inverting the cumulant is feasible, by first simulating $\sigma$ and then sampling $X$ conditionally on $\sigma$. Note that the restriction on the distribution of $L^{'}$ is not significant. Although the distribution of $X_t$ conditional on $\sigma$ is Stable, the unconditional distribution of $X_t$ does not have to be Stable.

The flexible marginal distribution and autocorrelation structure, as well as the computational efficiency and convergence properties of the simulation schemes in Section \ref{section:simulation} recommend the trawl process as a candidate for the stochastic volatility component. Thus we model $\sigma^2$ with a trawl process. Figure \ref{fig:volatility_modulated_trawl} shows such an example, where $X_t$ is conditionally Gaussian (which corresponds to $\alpha=2$ in the family of Stable distributions) on the volatility. More precisely, we use a Gaussian L\'evy basis $L$ and model $\sigma^2$ as a stationary trawl process with long memory and Inverse Gaussian marginal distribution. 
\begin{figure}[h]
  \centering
  \subfloat[Realisation of $\sigma^2$, modelled as a trawl process with long memory and Inverse Gaussian marginal]{\includegraphics[height = 5.75cm,keepaspectratio ]{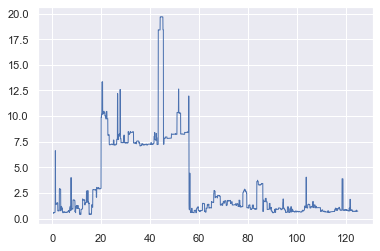}\label{fig:volatility_trawl}}
  \hfill
\subfloat[Realisation of the volatility modulated trawl process $X_t = \int_{A_t} \sigma(\bar{t}) L(\mathrm{d}\bar{t},\mathrm{d}\bar{x})$]{ \includegraphics[height = 5.75cm,keepaspectratio ]{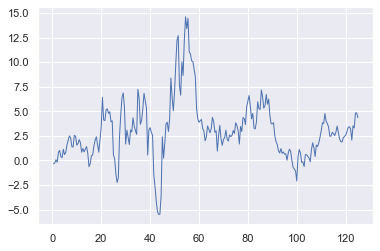}\label{fig:actual_volatility modulated trawl}}
 \caption{b) shows the simulation of a volatility modulated trawl process $X_t = \int_{A_t} \sigma(\bar{t}) L(\mathrm{d}\bar{t},\mathrm{d}\bar{x})$ with $L^{'}\sim \mathcal{N}(1,1)$ and $A =\{(s,x) : 0 < x < 0.25  e^{0.25 s}\}$ at times $\tau,\ldots,250\tau$, where $\tau = 0.5$. We use the slice partition method for L\'evy stable distributions. Conditionally on the values of $\sigma,$ $\int_S \sigma_s \mathrm{d}L_s \sim \mathcal{N}(\int_S \sigma(\bar{t}) \mathrm{d}\bar{t}\mathrm{d}\bar{x},\, \int_S \sigma^2(\bar{t}) \mathrm{d}\bar{t}\mathrm{d}\bar{x})$. To approximate these integrals we simulate the volatility on a fine equidistant grid. We model $\sigma^2$ as a trawl process with trawl set $B =\{(s,x): 0 < x < 0.5  (1-s)^{0.5} \}$ and L\'evy seed $L_\sigma^{'}\sim \text{Inverse Gaussian}(2,1)$, which we simulate from $-10$ to $250 \tau$ in steps of $\tilde{\tau} = 0.05$. A realisation of $\sigma^2$ is shown in a). The trawl functions are normalized such that the areas of $A$ and $B$ are $1$.}\label{fig:volatility_modulated_trawl}
\end{figure}

%% file: Extensions_to_ambit_field.tex
\section{Extensions to ambit field simulation}
\label{section:extensions_to_ambit_field_simulation}
So far we concentrated on the simulation of trawl processes, which amounts to evaluating the L\'evy basis, potentially modulated  by a kernel and stochastic volatility, over a collection of time-indexed trawl sets $A_t$. Note again that the trawl sets include an abstract spatial dimension in addition to the temporal dimension, which allows for a flexible joint distribution and autocorrelation function of the trawl process. A natural extension of the trawl process to spatio-temporal fields is the simple ambit field $Y$ given by
\begin{equation*}
    Y_t(x) = L\left(A_t(\mathbf{x})\right),
\end{equation*}
and more generally, the ambit field $Y$ given by
 \begin{equation*}
     Y_t(\mathbf{x}) = \int_{A_t(\mathbf{x})} K_{t,\mathbf{x}}\left(\bar{t},\bar{\mathbf{x}}\right) \sigma(\bar{t},\bar{\mathbf{x}}) L(\mathrm{d}\bar{t},\mathrm{d}\bar{\mathbf{x}}).
 \end{equation*}
 Ambit fields have already been used in turbulence and tumor growth modelling in \cite{barndorff2015intermittent}, and also outside spatio-temporal statistics, for example in electricity futures pricing \cite{barndorff-nielsen_benth_veraart_2014}.
 
We start Subsection \ref{subsection:simple_ambit_fields} by showing that the slice partition algorithm can be generalized to the simulation of simple ambit fields. As opposed to simulation via compound Poisson processes, which leads to an increased cost per simulation, the slice partition method can be implemented in a way such that the calculations required for higher accuracy only need to be performed once, before the simulation, leading to lower computational cost across simulations. Finally, we expand on the simulation of ambit fields in Subsection \ref{subsection:ambit_field_simulation}. 
\subsection{The slice partition method for simple ambit fields}
\label{subsection:simple_ambit_fields}
As in the trawl case, the autocovariance and autocorrelation structures for the simple ambit field $Y_t(\mathbf{x})= L(A_t(\mathbf{x}))$
 \begin{align*}
\Cov(Y_{t_1}(\mathbf{x_1}),Y_{t_2}(\mathbf{x_2})) &=  \mathrm{Leb}\left(A_{t_1}(\mathbf{x})\cap A_{t_2}(\mathbf{x_2})\right)  \Var(L^{'})\\ \Corr(Y_{t_1}(\mathbf{x_1}),Y_{t_2}(\mathbf{x_2})) &= \mathrm{Leb}\left((A_{t_1}(\mathbf{x_1})\cap A_{t_2}(\mathbf{x_2})\right) / {\mathrm{Leb}\left(A\right)} 
 \end{align*}
 and the cumulant transform $C\left(\theta,Y_t(\mathbf{x})\right) = \mathrm{Leb}(A) C(\theta,L^{'})$ present simple ambit fields as a tractable approach to modelling spatio-temporal data. We present the simulation algorithm for $d=2$ dimensions, as generalizing to more spatial dimensions is straightforward. The goal is then to simulate the simple ambit field $Y_t(x)$ at coordinates $\{(j\tau,ix): 1 \le j \le k_t, 1 \le i \le k_s\}$.

For ease of notation, let $A_{ij} = A_{j \tau}( i x)$ and define the lexicographic ordering $(m,n) \prec (i,j)$ if $m <  i$ or $m=i$ and $n< j,$ with equality when $m=i$ and $n=j.$ Intuitively, we order trawls from left to right and from bottom to top. 
For ease of presentation, assume that the trawl set is bounded; we relax this assumption in Appendix \ref{appendix:B2}. We say a set $S\subset \mathbb{R}^2$ is a minimal slice if it can be represented as
\begin{equation*}
    S = \bigcap_{(i,j) \in K} A_{ij}
\end{equation*}
for some indicator set $K \subset \mathbb{Z}^2$ and further $S \cap A_{ij} = \emptyset$ for any $(i,j) \not\in K.$ The ambit sets can be partitioned into disjoint minimal slices;  thus, to simulate the simple ambit field, it is enough to simulate all the minimal slices and keep track of which ambit sets $A_{ij}$ each slice belongs to. In this algorithm, we simulate L\'evy basis over trawls sets from left to right along each row, moving over rows from bottom to top. At step $(k,l),$ we simulate slices that belong to $A_{kl}$ and have empty intersection with trawls that are left and on the same row as $A_{kl},$ or bottom of $A_{kl}.$

Formally, define $I_t = \ceil{\frac{-T}{\tau}}$ and $I_s =  \ceil{\frac{\phi(0)}{x}}.$ Note that the sets $A_{ij}$ and $A_{i^{'}j^{'}}$ are disjoint whenever $|i - i^{'}| \ge I_s$ or $|j-j^{'}|\ge I_t.$ Thus each minimal slice $S$ can be represented by a minimal pair $(k,l) = \min{\{(i,j)\colon S \subset A_{ij}\}},$ where the minimum is understood in the sense of $\prec,$ and by an $I_s \times I_t$ indicator matrix $K^S$, where
\begin{equation*}
\left(K^S\right)_{ij} = 
\begin{cases}
    1 \text{ if } S \subset A_{i+k-1,j+l-1}, \\
    0 \text{ otherwise}.
    \end{cases}
\end{equation*}
Let $\mathcal{S}_{kl}$ be the set of minimal slices $S$ whose minimal pair is $(k,l).$ Then (see Figure \ref{fig:slices_coloured})
\begin{equation*}
    A_{kl} \ \backslash \bigcup_{(i,j) \prec (k,l)} A_{ij} = \bigcup_{S \in \mathcal{S}_{kl}} S.
\end{equation*}

\begin{figure}
    \centering
    \vspace{-2cm}
    \includegraphics[scale=0.7]{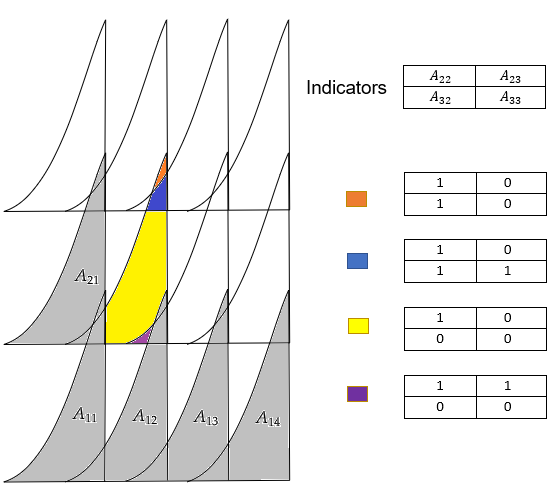}
    \caption{Assume that at step $(k,l)$ with $k=l=2,$ we have already simulated the minimal slices which are situated left and on the same row as $A_{22},$ or bottom of $A_{22}.$ These slices are shaded in gray and have been simulated at previous steps. 
    The parts coloured in orange, blue, yellow and purple illustrate the minimal slices of $\mathcal{S}_{22},$ together with the corresponding indicator sets $K^S$. These four minimal slices are the slices to be simulated at step $(2,2).$ Simulating the simple ambit field is then equivalent to simulating all the minimal slices $S$ and keeping track of the indicators $K^S.$}
    \label{fig:slices_coloured}
\end{figure}
By the translation invariance property of the grid of ambit sets $\{A_{ij}\}_{-\infty \le i,j \le \infty}$, \  $\mathcal{S}_{kl}$ and $\mathcal{S}_{k^{'}l^{'}}$ contain the same number of minimal slices $S,$ with the same Lebesgue measures and with indicators $K^S$ which are translated by $k-k^{'}$ and $l-l^{'}.$ Hence, to simulate the simple ambit field, it is enough to determine the Lebesgue measures and indicators of minimal slices in $S_{11}.$ We can identify the minimal slices $S$, approximate their Lebesgue measures and corresponding indicators $K^S$ via  Monte Carlo methods: sample points uniformly at random, keep track of the indicators $K^S,$ count how many points are in each minimal slice and divide the count by the total number of points to estimate the areas, as described in Algorithm \ref{algo:slice_estimation_compact_ambit_set}.

To account for boundary effects and simulate $\{A_{ij}\}_{\substack{1 \le i \le I_s \\ 1 \le j \le I_t }}$ exactly, we simulate the minimal slices in
 \begin{equation*}
     \bigcup_{\substack{-I_s+2 \le k \le k_s \\  -I_t+1 \le l \le k_t }} \mathcal{S}_{kl},
 \end{equation*}
 which means that we also simulate subsets of $A_{ij}$ with $-I_s+2 \le i \le 0, k_s < i \le k_s + I_s-1, \ -I_t+2 \le j \le 0,$ and $ k_t < j \le k_t+I_t-1.$ We discard the extra values. The pseudocode for the slice partition method for simple ambit fields is given in Algorithm \ref{algo:monotonic_ambit_field_simulation}.
Note that the calculations required for a higher accuracy can be performed ahead of the simulation. Regardless of the number of simulations, we only have to perform this procedure once, leading to an amortised computational cost across simulations. Figure \ref{fig:simple_ambit_field_simulations} displays two simulations of simple ambit fields. The above procedure can be generalized to unbounded trawls sets, as detailed in Appendix \ref{appendix:B2}. Just as in the slice partition algorithm for trawl processes, we have more slices to take into account and the computational complexity increases.
 \begin{figure}[h]
        \centering
        \begin{subfigure}[t]{0.32\textwidth}
         \includegraphics[width=\textwidth]{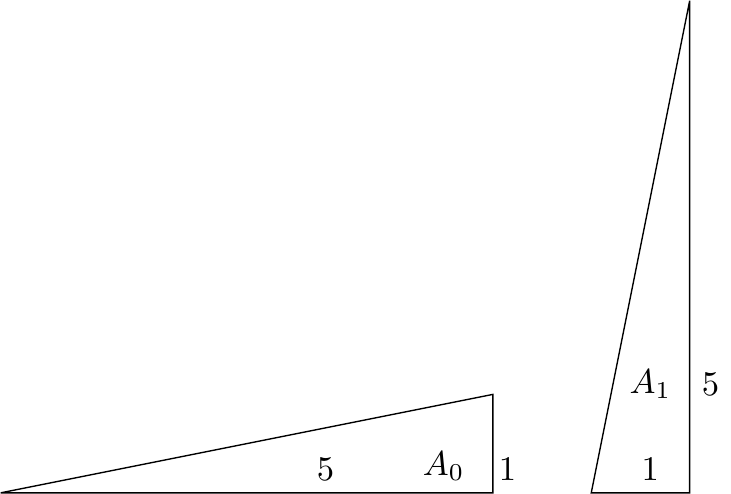}
        \caption{Ambit sets $A_0$ and $A_1$ with side lengths of $1$ and $5$ units.}
        \end{subfigure}
        \hfill
        \begin{subfigure}[t]{0.32\textwidth}
           \includegraphics[width=\textwidth]{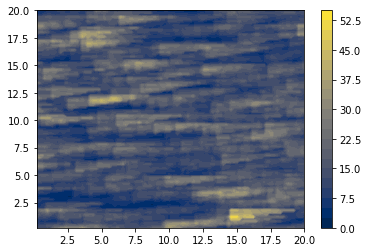}
            \caption{Simulation of a simple ambit field with ambit set $A_0$.}
        \end{subfigure}
        \hfill
        \begin{subfigure}[t]{0.32\textwidth}
           \includegraphics[width=\textwidth]{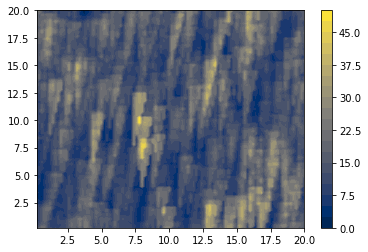}
            \caption{Simulation of a simple ambit field with ambit set $A_1$.}
        \end{subfigure}
        \caption{Simulations of two simple ambit fields with triangular ambit sets $A_0$ and $A_1$ and with $\tau= x = 0.2,$ $k_t=k_s=100$, $L^{'}\sim \text{Gamma}(2,3)$ and $N=10^8$ samples for the slice area estimation procedure. The time, space axes are the horizontal, respectively vertical ones. Note that for ambit sets such as $A_0$ and $A_1$, the minimal slices can be worked out by hand and the simulation is then exact. The values of the simulated ambit fields are given by the colorbars of b) and c). The shapes of the ambit sets induce qualitatively different  spatio-temporal autocorrelation structures.}
     \label{fig:simple_ambit_field_simulations}
\end{figure}
\input{algorithms/simple_ambit_field_slice_estimation_bdd_case.tex}
\input{algorithms/slice_partition_for_bounded_ambit_sets}
\subsection{Ambit field simulation}
\label{subsection:ambit_field_simulation}
We turn our attention to the general case of ambit fields $Y$ defined by 
\begin{equation*}
     Y_t(x) = \int_{A_t(x)} K_{t,x}\left(\bar{t},\bar{x}\right) \sigma(\bar{t},\bar{x}) L(\mathrm{d}\bar{t},\mathrm{d}\bar{x}),
 \end{equation*}
where $K$ is a deterministic kernel and $\sigma$ is a stochastic volatility field. We aim to simulate $Y_t(x)$ at coordinates $\{(j\tau,ix): 1 \le j \le k_t, 1 \le i \le k_s\}$. Assume $\sigma$ and $L$ are independent. Then the second order structure and cumulant transform follow from Equations \eqref{eq:conditional_mean}-\eqref{eq:conditional_cumulant}. Similarly to Subsection \ref{subsection:volatility_modulated_trawls}, we advocate for the use of a simple ambit fields for the stochastic volatility term $\sigma$, which can be simulated efficiently. Conditional on the values of $\sigma$, we can simulate $Y$.

Analogous to L\'evy semistationary processes, we can choose ambit sets with simple geometries, such as unbounded rectangles $A_t(x) = \{(\bar{t},\bar{x}): \bar{t} < t, x<\bar{x} < x+1\}$, and control the joint structure solely through the kernel. In this case, we can identify the minimal slices by hand. In general, if we use a more complicated ambit set and a simple kernel, this is not possible and we identify minimal slices by their indicator matrices, as in Subsection \ref{subsection:simple_ambit_fields}. In both cases, we separate $L$ into $L_g$ and $L_j$ and simulate their contributions independently, conditionally on $\sigma$. The jump part $Y_t(x) = \int_{A_t(x)} K_{t,x}\left(\bar{t},\bar{x}\right) \sigma(\bar{t},\bar{x}) L_j(\mathrm{d}\bar{t},\mathrm{d}\bar{x})$ can be approximated up to small jump truncation by a discrete sum, as in Section \ref{section:extensions_to_VMKWTP} and only requires access to the values of the kernel $K$ and volatility $\sigma$. The Gaussian part $Y_t(x) = \int_{A_t(x)} K_{t,x}\left(\bar{t},\bar{x}\right) \sigma(\bar{t},\bar{x}) L_g(\mathrm{d}\bar{t},\mathrm{d}\bar{x})$ requires computing the means and covariances
\begin{align*}
   & \ev\left[L^{'}\right]\int_{A_t(x)} K_{t,x}(\bar{t},\bar{x})\sigma(\bar{t},\bar{x}) \mathrm{d}\bar{t}\mathrm{d}\bar{x},\\
    &\Var{\left(L^{'}\right)}\int_{A_{t_1}(x_1)\cap A_{t_2}(x_2)} K_{t_1,x_1}(\bar{t},\bar{x})K_{t_2,x_2}(\bar{t},\bar{x})\sigma^2(\bar{t},\bar{x})\mathrm{d}\bar{t}\mathrm{d}\bar{x},
\end{align*} which can be approximated similarly to Algorithm \ref{algo:slice_estimation_compact_ambit_set}, by simulating points equidistantly or uniformly at random, evaluating the kernel and volatility at these points and taking the mean. 
 
 Further simplifications are possible in particular cases. If the kernel satisfies $K_{t,x}(\bar{t},\bar{x}) = g(\bar{t}-t,\bar{x}-x)$ for some square integrable $g$ and $\Var{\left(L^{'}\right)}$ is finite, the Fourier approximation methodology from Subsection \ref{subsection:t-dependent-kernels} can be used to remove the dependence of the kernel on $t$ and $x$. Similarly, if $L^{'}$ has a stable distribution, it is enough to first simulate the volatility field $\sigma$ and then approximate the parameters of the conditional distribution of $\int_S K_{t,x}(\bar{t},\bar{x}) \sigma(\bar{t},\bar{x}) L\left(\mathrm{d}\bar{t},\mathrm{d}\bar{x}\right)$ for all minimal slices $S$.
 
\section{Conclusion}
After presenting the elementary properties of L\'evy bases in Section \ref{section:The Trawl process framework}, we introduced three algorithms for the simulation of trawl processes in Section \ref{section:simulation}. We developed their theoretical error analysis, discussed their computational complexity and provided easily adaptable computer implementations. Further, we investigated the effects of approximating a L\'evy basis through grid discretization and small jump truncation. While the three simulation schemes were initially presented in the trawl processes framework, we showed in Sections \ref{section:extensions_to_VMKWTP} and \ref{section:extensions_to_ambit_field_simulation} that they are directly applicable to the more general settings of kernel-weighted, volatility modulated trawl processes, simple ambit fields and ambit fields. Moreover, we showed in Subsection \ref{subsection:t-dependent-kernels} that simulation schemes previously studied in the literature could be combined with our methods for decreased computational cost. All of the above enable the implementation of high-accuracy simulation studies and simulation-based inference and bring Ambit Stochastics closer to widespread use when modelling real-world data.

\textbf{Acknowledgements}

We would like to thank Dan Crisan for constructive discussions and comments on earlier versions of the manuscript. Dan Leonte acknowledges support from the EPSRC Centre for Doctoral Training in Mathematics of Random Systems: Analysis, Modelling and Simulation (EP/S023925/1).

%% file: algorithms/simple_ambit_field_slice_estimation_bdd_case.tex
\begin{algorithm}[h]
   \caption{Slice estimation for bounded, monotonic ambit sets}\label{algo:slice_estimation_compact_ambit_set}
  \begin{algorithmic}[1]
        \Require Trawl function $\phi$ and $T<0$ with $\phi(T)=0;$ number of samples $N$ to be used in the estimation; sampler $U(a,\,b,\,n)$ which returns an $n$ dimensional vector of iid samples from the uniform distribution on $[a,b]$
    \Ensure Hash table $H[\text{key:value}]$ mapping the keys, $I_s \times I_t$ indicator matrices of minimal slices, to their correspoding values, given by the estimated Lebesgue measures of the minimal slices.
\Function{SliceEstimation}{$U,\,N,\,\phi,\,T,\,I_t,\,I_s,\,\tau,\,x$}
  \State $t \sim U\left(0,\tau,N\right)$
      \State $x \sim U\left(x,x+\phi(0),N\right)$ 
      \For{$l \in \{1,\ldots,N\}$}
        \If{$x[l] > \phi(t[l]-\tau)$\AND
        $0 < x[l] - x < \phi(t[l]-\tau)$ }
        \Comment{\parbox[t]{.38\linewidth}{Exclude points which belong to $A_{01}$ and points which do not belong to $A_{11}$}}
        \State $I \gets \text{zeros}(I_s,I_t)$ \Comment{Indicator matrix for a minimal slice}
        \For{$i \in \{1,\ldots,I_s\}$}
         \For{$j \in \{1,\ldots,I_t\}$}
         \State \algorithmicif \ $x[l]- ix < \phi(t[l]-j\tau)$ \algorithmicthen \ $I[i,j] \gets 1$ \Comment{Check if $(t[l],x[l]) \in A_{ij}$}
        \EndFor
        \EndFor
        \State \algorithmicif\ $I \in \text{keys}(H)$ \algorithmicthen\ $H[I] \gets H[I] + 1$ \algorithmicelse\ $H[I] \gets  1$
        \EndIf
        \EndFor
  \For{$I \in \text{keys}(H)$}
  \State $H[I] \gets  H[I] \tau \phi(0) / N$ \Comment{Approximate the area of a minimal slice $S$ with indicator $I$}
  \EndFor
 \Return $H$
 \EndFunction
  \end{algorithmic}
\end{algorithm}

%% file: algorithms/slice_partition_for_bounded_ambit_sets.tex
\begin{algorithm}[h]
  \caption{Slice partition for bounded, monotonic ambit sets}\label{algo:monotonic_ambit_field_simulation}
  \begin{algorithmic}[1]
      \Require Trawl function $\phi$ and $T<0$ with $\phi(T)=0;$ number of samples $N$ to be used in the SliceEstimation procedure; sampler $U(a,b,n)$ which returns an $n$ dimensional vector of iid samples from the uniform distribution on $[a,b];$ sampler $T(\textrm{area})$ which returns one sample from a collection of independent random variables with the same law as $L(A),$ where $\mathrm{Leb}(A) = \textrm{area};$ number of ambit sets to be simulated $k_t$, $k_s$, distance $\tau$,$x$ between ambit sets.
      \Ensure $k_s \times k_t$ matrix $Y$ containing the values of the simulated simple ambit field at coordinates $\{(j\tau,ix): 1 \le j \le k_t, 1 \le i \le k_s\}$
      \Function{SlicePartition}{$U,\,T,\,N,\,\phi,\,k_t,\,k_s,\,\tau,\,x$}
      \State $I_t, I_s \gets \ceil{\frac{-T}{\tau}},   \ceil{\frac{\phi(0)}{x}} $
  \State $ H \gets \Call{SliceEstimation}{U,\,N,\,\phi,\,I_t,\,I_s,\,\tau,\,x}$
  \State $Y \gets \text{zeros}(k_s+2I_s-2,\,k_t+ 2I_t-2).$
  \Comment{Corresponding to the matrix  $L\left(A_{ij}\right)_{\substack{-I_s+2 \le i \le k_s+I_s-1 \\ -I_t+2 \le j \le k_t + I_t-1 }}$}
  \For{$k \in \{1,\ldots,k_s+I_s-1\}$}
\For{$l \in \{1,\ldots,k_t+I_t-1\}$}
\For {$I \in \text{keys}(H)$}
\State $c \gets T(\text{area})$\Comment{Simulate $L(S)$ for each $S \in \mathcal{S}_{kl}$} 
 \State $Y[k:k+I_s-1,l:l+I_t-1] \pluseq c I.$ 
  \EndFor
 \EndFor
  \EndFor
  \Return $Y[I_s:I_s + k_s-1,I_t:I_t+k_t-1].$  \Comment{Corresponding to the matrix $L\left(A_{ij}\right)_{\substack{1 \le i \le k_s \\  1 \le j \le k_t}}$}
\EndFunction
  \end{algorithmic}
\end{algorithm}

%% file: proofs.tex
\section{Proofs}
\label{section:proofs}
\begin{proof}[Proof of Theorem \ref{grid_discretisation_thm}]
\label{proof_of_grid_discr_thm}
Let $T_n \to -\infty, \Delta_n = (\Delta^n_t,\Delta^n_x) \to 0.$ Let $Q_n$ be the union of cells in the grid $G_{T_n,\Delta_n}$ which are also contained in $A$
\begin{equation*}
    Q_n = \{g \in G_{T_n,\Delta_n}  :  g \subset A  \},
\end{equation*}and let $U_n = A \backslash Q_n.$ Then $L_{T_n,\Delta_n}(A) = L(Q_n).$ 
Note that for each $n,$ there is some $N$ such that for any $k>N, Q_n \subset Q_k;$  thus we can extract a nested subsequence $\{Q_{n_k}\}_{_k}.$ Then $\cup_k Q_{n_k} = A$ and  \begin{equation}\mathrm{Leb}\left(Q_{n_k}\right)\to \mathrm{Leb}\left(A\right) \text{ as } k \to \infty.
\label{eq:nested_convergence}
\end{equation}
By \eqref{eq:nested_convergence} and again by the fact that for each $n$ there is some $N$ such that for any $k>N, Q_n \subset Q_k,$ we obtain that
that $\mathrm{Leb}\left(Q_n\right) \to \mathrm{Leb}(A)$, and consequently $\mathrm{Leb}\left(U_n\right) \to 0$ as $n \to \infty.$ By \eqref{eq:relate_cumul_L_with_cumul_levy_seed}, we obtain $  C(\theta,L(U_n)) = C(\theta,L^{'}) \mathrm{Leb}(U_n) \to 0 \text{ as } n \to \infty.$ Thus $L(U_n) \to 0$ in distribution, hence in probability, and $L(Q_n) \to L(A)$ in probability as $n \to \infty.$

For the second part of this proof, assume that $\Var{L^{'}}$ is finite. By \eqref{eq:mean_L(A)},\eqref{eq:var_L(A)} and since $\mathrm{Leb}(U_n) \to 0$, we have that $\Var\left(L(U_n)\right) = \ev[L^{'}]^2 \mathrm{Leb}(U_n)^2 + \Var(L^{'}) \mathrm{Leb(U_n)} \to 0$ and $L(Q_n) \to L(A)$ in $\mathcal{L}^2$ as $n \to \infty.$ We now provide a bound on the MSE of the approximation. Let $a_n$ be the  number of grid cells that are misplaced for the trawl set $A$, i.e. the number of cells which are not counted as part of $A$ despite having non-empty intersection with $A$. Since $\phi$ is increasing and continuous, by a counting argument we have that $a_n \le T_n/\Delta_t + \phi(0)/\Delta_s,$ hence $\mathrm{Leb}(U_n) = \mathrm{Leb}(U_n \cap \{t < T_n\}) + \mathrm{Leb}(U_n \cap \{T_n < t <0\}) \le \int_{-\infty}^{T_n} \phi(u) \mathrm{d}u + a_n \Delta_t\Delta_x = \int_{-\infty}^{T_n} \phi(u) \mathrm{d}u + T_n \Delta_s + \phi(0) \Delta_t  \eqdef C_n$ and $\ev \left[\left(X^n- L(A)\right)^2 \right]  = \ev \left[L\left(U_n\right)^2\right] \le C_n^2 \ \ev\left[L^{'}\right]^2 + C_n \Var{L^{'}}.$
\end{proof}
The above MSE bound is not sharp, in the sense that $T_n$ and $\Delta^n_s$ must be chosen such that $T_n \Delta^n_s \to 0$ in order for the bound to be meaningful. This is because we upper-bound $\mathrm{Leb}(U_n)$ by the the number of cells in $G_{T_n,\Delta_n}$ which have non-empty intersection with $A,$ i.e. $a_n$, timesed by the area of the cells, which is $\Delta_t\Delta_x$. Whereas taking $T_n$ to be negative and of large absolute value gives a large $a_n,$ these cells have less and less overlap with $A,$ a fact which we do not account for in the MSE bound.
\begin{proof}[Proof of Theorem \ref{theorem:uniform_convergence}] Note that the jumps with absolute value larger than $1$ can be simulated exactly, hence it is enough to deal with the jumps which have values less than $1$. Fix $\epsilon_n \downarrow 0$ with $ \epsilon_1 =  1$ and let $\left(Z^n\right)_{t \ge 0}$ given by
\begin{equation*}
Z^n_t = L_j^{\epsilon_n}(A_t)  = \int_{A_t} \int_{\abs{y} \in (\epsilon_{n+1},\epsilon_n)} y \ \mathrm{d} (N- \nu)(\mathbf{z},y),
\end{equation*}
where $\mathbf{z} = (t,\mathbf{x}).$ Note that $\ev\left[Z^n_t\right] =0$ and $\ev\left[\left(Z^n_t\right)^2\right] = \mathrm{Leb}(A) \int_{\abs{y} \in (\epsilon_{n+1},\epsilon_n)} y^2 l(\mathrm{d}y) \le \mathrm{Leb}(A) \int_{-1}^1 y^2 l(\mathrm{d}y)$, which is finite. Thus $Z^n$ are independent, zero-mean, square integrable random c\`adl\`ag stochastic processes whose marginals $Z^n_t$ satisfy the assumptions of Lemmas $20.2,20.4$ and $20.5$ from \cite{ken1999levy}. The conclusion follows by noticing that the above three results can still be used with our definition of $Z^n$ in the proof of Lemma $20.6$ from \cite{ken1999levy}.

We proved that $X^\epsilon \to X$ pathwise uniformy on compacts for a jump L\'evy basis $L_j.$ We show next that $t \to  L_g(A_t)$ is continuous a.s. Since continuous functions are uniformly continuous on compacts, the convergence holds for any L\'evy basis. 
\begin{lemma}\label{remark:cts_paths}
The trawl process given by $X_t = L_g(A_t),$ where $L_g$ is a Gaussian L\'evy basis, has H\"older continuous paths with exponent $\alpha \in (0,1/2).$
\end{lemma}
\begin{proof}
Let $\rho$ be the autocorrelation function of $X$. Note that $X_t - X_{t+h} = L(A_t \backslash A_{t+h}) - L(A_{t+h} \backslash A_t)$ and that $L(A_t \backslash A_{t+h}), L(A_{t+h} \backslash A_t)$ are iid with law $\mathcal{N}\left(\mathrm{Leb}(A) \rho(h) \mu,\, \mathrm{Leb}(A) (1-\rho(h)) \sigma^2\right).$ Thus $X_t - X_{t+h} \sim \mathcal{N}\left(0,2 \ \mathrm{Leb}(A) (1-\rho(h)) \sigma^2\right)$ and for any $p \in \mathbb{Z}_{\ge 1},$ we have $\ev\left[\left|X_t-X_{t+h}\right|^p\right] = C  |1-\rho(h)|^{p/2},$ where $C = \left(2 \ \text{Leb}(A)) \sigma^2\right)^{p/2} \frac{\Gamma(\frac{p+1}{2})}{\sqrt{\pi}}.$ Note that the trawl function $\phi\colon (-\infty,0] \to \R_{\ge 0}$ is assumed to be continuous and increasing, hence $\rho$ is $C^1$ and $\abs{\rho^{'}(h)} \le \phi(0)$. By the Mean Value Theorem and by \eqref{eq:autocor_intro}, we have that 
\begin{equation*}
    \ev\left[\left|X_t-X_{t+h}\right|^p\right] = C | \rho'(\xi(h))|^{p/2} h ^{p/2} \le C \phi(0)^{p/2} \mathrm{Leb}(A)^{p/2} h^{p/2} = \tilde{C} h^{p/2},
\end{equation*}
for some $\xi(h) \in (0,h)$ and $\tilde{C} = C \phi(0)^{p/2} \mathrm{Leb}(A)^{p/2}.$ We conclude by Kolmogorov's continuity theorem that $X_t$ has H\"older continuous paths with exponent $\alpha \in (0,1/2).$
\end{proof}
We have thus proved that $X^\epsilon \to X$ pathwise uniformly on compacts regardless of the L\'evy basis. 
\end{proof}
\begin{remark}[Uniform convergence for kernel-weighted trawl processes]
\label{remark:uniform_convergence_for_KW}The result of Theorem \ref{theorem:uniform_convergence} can easily be extended to kernel-weighted trawl processes over finite Lebesgue measure trawl sets $A_t = A + (t,0)$. Assume that $t \to \int_{A_t} K^2(\bar{t},\bar{\mathbf{x}}) L_j(\mathrm{d}\bar{t},\mathrm{d}\bar{\mathbf{x}})$ is bounded on compacts $[0,T]$ and define 
\begin{equation*}
Z^n_t =  \int_{A_t} \int_{\abs{y} \in (\epsilon_{n+1},\epsilon_n)} K(\bar{t},\bar{\mathbf{x}}) y \ \mathrm{d} (N- \nu)(\mathbf{z},y),
\end{equation*}
where $\mathbf{z} = (\bar{t},\bar{\mathbf{x}}) \in S$. Then $Z^n$ are c\`adl\`ag, $\ev\left[Z^n_t\right] =0 $ and  \begin{equation*}
\ev\left[\left(Z^n_t\right)^2\right] = \int_{A_t} \int_{\abs{y} \in (\epsilon_{n+1},\epsilon_n)} K^2(\bar{t},\bar{\mathbf{x}}) y^2 \  l\mathrm{d}(y) \mathrm{Leb}(\mathrm{d}\mathbf{z}) =  \int_{A_t}  K^2(\bar{t},\bar{\mathbf{x}})  \mathrm{Leb}(\mathrm{d}\mathbf{z}) \int_{\abs{y} \in (\epsilon_{n+1},\epsilon_n)} y^2 \  l\mathrm{d}(y),
\end{equation*}
which is finite and converges to $0$ as $n \to \infty$. The same proof as above can be reused to prove the uniform convergence on compacts of $\int_{A_t} K^2(\bar{t},\bar{\mathbf{x}}) L_j^\epsilon(\mathrm{d}\bar{t},\mathrm{d}\bar{\mathbf{x}})$ to $\int_{A_t} K^2(\bar{t},\bar{\mathbf{x}}) L_j(\mathrm{d}\bar{t},\mathrm{d}\bar{\mathbf{x}})$. If $K$ is time dependent, we also require continuity of $t \to K_t(\bar{t},\bar{x})$ for each fixed pair $(\bar{t},\bar{x}),$ so that $t \to \left(Z^n\right)_t$ still has  c\`adl\`ag paths.
\end{remark}
\begin{proof}[Proof of Theorem \ref{theorem:skorohod}] By the Poisson random measure representation of $L_j$ in Theorem \ref{thm:levy_ito_decomp_for_levy_bases}, we have that $t \to L_j(A_t)$ is c\`adl\`ag. Further, by Lemma \ref{remark:cts_paths}, we obtain that $t \to L_g(A_t)$ is continuous, hence  $t \to L(A_t)$ is c\`adl\`ag. The theorem follows then from Lemma \ref{lemma_j1_convergence}.
\end{proof}

%% file: appendix_file.tex
\begin{appendices}
\section{Background Material}
\label{appendix:A}
\begin{theorem}[L\'evy-It\^o decomposition of L\'evy bases: Theorem 4.5 in \cite{pedersen2003levy}]
\label{thm:levy_ito_decomp_for_levy_bases}
Let $L$ be a L\'evy basis on $S \subset \R^d$ with characteristic triplet $(\xi,\,a,\,l)$, where $\xi \in \R,\, a \in \R_{\ge 0}$ and $l$ is a L\'evy measure on $\R$. Let $\nu = \mathrm{Leb} \otimes l$ be the product measure on $S \times \R,$ where $\mathrm{Leb}$ is the Lebesgue measure on $S.$ 
Then there exists a decomposition of $L$ into independent L\'evy bases: a Gaussian part $L_g$ with characteristic triplet $(\xi,\,a,\,0)$ and a jump part $L_j$ with characteristic triplet $(0,\,0,\,l)$ such that for any $A \in B_{\mathrm{Leb}}(S),$ the following holds a.s.
\begin{equation*}
    L(A) = L_g(A) + L_j(A).  
\end{equation*}
Moreover, there exists a Poisson random measure $N$ on $\mathcal{B}_{\mathrm{Leb}} (S)\times \R$ with intensity measure $\nu$ such that
\begin{equation}
    L_j\left(A\right) = \int_A \int_{(-1,1)} y \ \mathrm{d} (N- \nu)(\mathbf{z},y) + \int_A \int_{\R \backslash(-1,1)} y \ \mathrm{d} N(\mathbf{z},y),
    \label{eq:poisson representation}
\end{equation}
where $\mathbf{z} = (t,\mathbf{x}) \in S$ and $y \in \R.$ 
\end{theorem}
Let $\mathrm{D}_{\mathbb{R}}[0, 1]$ be the space of real-valued c\`adl\`ag functions on $[0, 1].$ 
\begin{definition}[Skorokhod's J1 topology]
Define the family of time changes $\Lambda$ to be the set of all increasing homeomorphisms $\lambda \colon [0,1] \rightarrow [0,1].$ The J1 topology on $\mathrm{D}_{\mathbb{R}}[0,1]$ is induced by the following distance 
\begin{equation*}
    d_{J_1}(x, y)=\inf _{\lambda \in \Lambda}\left(\sup _{t \in[0,1]}|\lambda(t)-t|+\sup _{t \in[0,1]}|x(\lambda(t))-y(t)|\right).
\end{equation*}
\label{def:skorohod_top}
\end{definition}
\begin{lemma}[Uniform c\`adl\`ag regularity]
Given $f \in \mathrm{D}_{\mathbb{R}}[0, 1]$ and $\varepsilon>0,$ there exists a partition $\pi_{\varepsilon}=\left\{t_{0}, \ldots, t_{k(\varepsilon)}\right\}$ of $[0, 1]$ such that
\begin{equation*}
\sup _{r, s \in\left[t_{1}, t_{\mathfrak{t}+1}\right)}\left|f(s)-f(r)\right| \leq \varepsilon \quad \text { for } \quad i=0, \ldots, k(\varepsilon)-1.
\end{equation*}
\label{thm:uniform_cadlag_reg}
\end{lemma}
\begin{lemma}\label{lemma_j1_convergence}
Let $X,X^n$ be stochastic processes with sample paths in $\mathrm{D}_{\mathbb{R}}[0,1].$  Assume that $X^n$ is constant on the intervals $[0,\frac{1}{2^n}),\ldots,[\frac{2^n-1}{2^n},1)$ and that
\begin{equation}
\max_{i \in D_n} \left|X^n_i - X_i\right| \to 0 \text{ a.s. as } n \to \infty \label{eq:1st condition}
\end{equation}
where $\mathcal{D}_n = \{0,\frac{1}{2^n},\ldots,\frac{2^n-1}{2^n}, 1\}.$ Then $X_n \to X$ a.s. in Skorokhod's J1 topology. 
\end{lemma}
\begin{proof}[Proof of Lemma \ref{lemma_j1_convergence}]
Define \begin{align*}
d_n^{-}(t)\colon & [0,1] \to [0,1],\ t \to \frac{\floor{2^n t}}{2^n},\\
d_n^{+}(t)\colon & [0,1] \to [0,1], \ t \to \min{\left(\frac{\floor{2^n t} +1 }{2^n},1\right)},
\end{align*} where $\floor{\cdot}$ is the floor function. 
Fix paths $X(\omega),\,X^n(\omega)$ and $\varepsilon>0.$ As there is no risk of confusion, we omit the $\omega$ in this proof. By the uniform c\`adl\`ag regularity property from Lemma \ref{thm:uniform_cadlag_reg}, there exists a partition $\pi_{\varepsilon}=\left\{0=t_{0}, \ldots, t_{k(\varepsilon)}=1\right\}$ of $[0, 1]$ such that
\begin{equation}
\sup _{r, s \in\left[t_{1}, t_{\mathfrak{t}+1}\right)}\left|X_s-X_r\right| \leq \varepsilon/3 \quad \text { for } \quad i=0, \ldots, k(\varepsilon)-1.
\label{lemma_inside_proof_uniform_cadlag_Reg}
\end{equation}
To show convergence in Skorokhod's J1 topology, we need to find an appropriate family of increasing homeomorphisms $\lambda^n \colon [0,1] \rightarrow [0,1]$ and a positive integer $N$ such that
\begin{equation}
\sup _{t \in[0,1]}\left|X^n_t -  \left({X \circ \lambda^n}\right)_t\right| + \sup _{t \in[0,1]}\left|\lambda^n_t-t\right| < \varepsilon,
\label{proof:skorohod}
\end{equation}
for $n>N.$ Let $N_1 > 2-\log_2{\left(\min_{0\le i \le k(\epsilon)-1}{(t_{i+1}-t_i)}\right)} .$ Then for $n>N_1,$ any two consecutive jump times $t_i,t_{i+1}$ are separated by at least 
$4/2^n$, hence $d_n^{+}(t_i)$ and $d_n^{-}(t_{i+1})$ are separated by at least $2/2^n$ and the intervals $I_{n,i} \defeq \left[d_n^{-}(t_i)-\frac{1}{2^n},d_n^{+}(t_i)+\frac{1}{2^n}\right]$ are disjoint. Let 
\begin{equation*}
    \mathcal{D}_n^1 = \bigcup_{i=1}^{k_\epsilon-1} I_{n,i},\ \mathcal{D}_n^2 = [0,1] \backslash \mathcal{D}_n^1.
\end{equation*}
On $\mathcal{D}_n^2,$ let $\lambda^n = \mathrm{Id};$ on each of the intervals $\left[d_n^{-}(t_i)-\frac{1}{2^n},d_n^{+}(t_i)\right],$ if $t_i \in \mathcal{D}_n,$ let $\lambda^n = \mathrm{Id},$ and if $t_i \notin \mathcal{D}_n,$ let 
\begin{equation*}\lambda^n(t) = \begin{cases}
d_n^{-}(t_i)-\frac{1}{2^n} + 2^{n-1}\left(t-d_n^{-}(t_i)+\frac{1}{2^n}\right)\left(t_i-d_n^{-}(t_i)+\frac{1}{2^n}\right)  &\text{ if }   d_n^{-}(t_i) - \frac{1}{2^n} \le t \le d_n^{+}(t_i), \\
t_i + 2^n(d_n^{+}(t_i)+\frac{1}{2^n} -t_i)\left(t-d_n^{+}(t_i)\right)  &\text{ if } d_n^{+}(t_i) \le t \le d_n^{+}(t_i)+\frac{1}{2^n},
\end{cases}
\end{equation*}
which is the piece-wise linear interpolation between
$\lambda^n\left(d_n^{-}(t_i)-\frac{1}{2^n}\right)=d_n^{-}(t_i) -\frac{1}{2^n}, \lambda^n(d_n^{+}(t_i)) = t_i$ and $\lambda^n(d_n^{+}(t_i)+\frac{1}{2^n}) = d_n^{+}(t_i)+\frac{1}{2^n}$.
By construction, $\sup_{t \in I_{n,i}} \left|\lambda^n_t-t\right| = \left|\lambda_n\left(d_n^{+}(t_i)\right) - d_n^{+}(t_i)\right| < \frac{1}{2^{n-1}},$ hence
\begin{equation*}
\sup_{t \in [0,1]} |\lambda^n(t)-t| = \max_{1 \le i \le k(\epsilon)-1} \sup_{t \in I_{n,i}} \left|\lambda^n_t-t\right| < \frac{1}{2^{n-1}} < \varepsilon/3,
\end{equation*}
for $n> 1-\log_2{\left(\varepsilon/3\right)}.$ Thus it is enough to show that for large $n,$ the following holds
\begin{equation*}
    \sup _{t \in[0,1]}\left|X^n_t -  {\left(X \circ \lambda^n\right)}_t\right| < \frac{2\varepsilon}{3}.
\end{equation*}
Since $\left|X^n_t - \left(X\circ \lambda^n\right)_t\right| \le \left|X^n_{d_n^{-}\left(t\right)} - X_{d_n^{-}\left(t\right)}\right| + \left|X_{d_n^{-}\left(t\right)} - \left(X\circ \lambda^n\right)_t\right|,$ we have that
\begin{equation*}
\sup_{t \in [0,1]}  |X^n_t - \left(X\circ \lambda^n\right)_t| \le \max_{d \in \mathcal{D}_n}| X^n_{d} - X_{d}| + \sup_{t \in [0,1]}|X_{d_n^{-}\left(t\right)} - \left(X\circ \lambda^n\right)_t|.  
\end{equation*}
By construction, we have that $\lambda^n(t)$ and $d_n^{-}(t)$ are in the same interval $[t_i,t_{i+1}),$ hence by inequality \eqref{lemma_inside_proof_uniform_cadlag_Reg}
\begin{equation*}
    \sup_{t \in [0,1]}|X_{d_n^{-}\left(t\right)} - \left(X\circ \lambda^n\right)_t| < \varepsilon/3,
\end{equation*}
for any $n>N_1.$
By relation \eqref{eq:1st condition}, there is some $N_2$ such that for any $n>N_2$
\begin{equation*}
    \max_{d \in \mathcal{D}_n} |X^n_{d} - X_{d}| < \varepsilon/3.
\end{equation*}
Thus inequality \eqref{proof:skorohod} is satisfied for any $n>\max{\{N_1,\,N_2,\,1-\log_2{\varepsilon/3}\}},$ which finishes the proof. Note that the lemma stays true for any nested sets $\mathcal{D}_1\subset \ldots \subset \mathcal{D}_n \subset \ldots \subset [0,1]$ with $\cup_n \mathcal{D}_n$ dense in $[0,1].$
\end{proof}
\section{Computational methods}
\label{appendix:B}
\subsection{Implementation of the convolution step in the slice partition algorithm}
\label{appendix:B1}
\begin{definition}[Convolution]
\label{def:convolution}
Let $Y$ be an $n \times k$ matrix and $F$ an $n \times n$ matrix, with $k>n.$ Let $Y[:,l:m]$ be the matrix obtained by removing columns $1,\ldots,l-1$ and $m+1,\ldots,k$ from $Y$ and let $\odot$ denote the sum of the component-wise product of two matrices. Define the convolution of $Y$ and $F$ with stride $1$ to be the $k-n$ dimensional vector $Y*F$ with $l^\text{th}$ entry given by $Y[:,l : l+n] \odot F$
\begin{equation*}
    [Y*F]_l = \sum_{i=1}^{n} \sum_{j=1}^{n} Y_{i,j+l-1} F_{ij}, \ 1 \le l \le k-n.
\end{equation*}
\end{definition}
The convolution step in Algorithm \ref{algo:slice_partition_trawl_process_convolution} can be implemented by repeated matrix multiplication, which requires $I^3k$ operations, or by the Fast Fourier transform. Nevertheless, when the values of $I$ and $k$ are close, the overhead associated with converting to and from the Fourier space can offset the speedup obtained by using the Fast Fourier transform. To account for such issues, we propose a new implementation of the convolution step, which takes advantage of the special form of the $I \times I$ filter
\begin{equation*}
  F = \begin{pmatrix}
    0 & 0 & \dots & 0 & 1 \\
    0 & 0 & \dots & 1 &1 \\
    \vdots & \vdots & \udots & \vdots & \vdots \\
    0 & 1 & \dots & 1 &1 \\
    1 & 1 & \dots & 1 & 1
  \end{pmatrix}.
\end{equation*}
With $Y_{*}$ defined below,
\begin{equation*}
Y_{*} = \begin{pmatrix}
        L\left(S_{11}\right)  & \ldots &  L\left(S_{1k}\right) \\
        L\left(S_{21}\right)  &  \ldots &  L\left(S_{2k}\right)\\
         \vdots        &    & \vdots  \\
        L\left(S_{I1}\right) &  \ldots &  L\left(S_{Ik}\right)
\end{pmatrix}.
\end{equation*}
the simulation scheme for trawls with finite decorrelation time is given by Algorithm \ref{algo:slice_partition_trawl_process_implementation} and requires $2 I k  - I^2/2 - I/2$ additions: $k(I-1)$ in Step $7$ and $I(k-I) + I(I-1)/2$ in Steps $9-11.$
\input{algorithms/slice_partition_unbounded_monotonic_trawls} The procedure in algorithm
\ref{algo:slice_partition_trawl_process_implementation} applies directly to trawls with infinite decorrelation time with the only mention that $Y_{*}$ is now given by the upper triangular $k$ by $k$ matrix
\begin{equation*}
\begin{pmatrix}
L\left(S_{11}\right)  & \ldots & L\left(S_{1k}\right) & L\left(S_{1k}\right) \\
        L\left(S_{21}\right)  &  \ldots &  L\left(S_{2,k-1}\right) & 0 \\
         \vdots &    \udots   &  \vdots  & \vdots  \\
        L\left(S_{k1}\right) &  \ldots &  0 & 0
\end{pmatrix}.
\end{equation*}
Similarly to the finite decorrelation time, the areas corresponding to the slices in the above matrix are available as integrals of the trawl function $\phi$ 
\begin{equation}
\label{areas_s_ij_non_compact_ambit_set}
\begin{pmatrix}
              a_1  &  a_1 -a_2  & a_1 -a_2  &    a_1 -a_2 &  \ldots  & a_1 -a_2 & a_1 -a_2 & a_1\\
              a_2  &  a_2 -a_3  & a_2 -a_3  &    a_2 -a_3&  \ldots  & a_2 -a_3 & a_2       & 0   \\
              a_3  &  a_3 -a_4  & a_3 -a_4  &    a_3 -a_4&  \ldots  & a_3       & 0          & 0   \\
                    &              &             &              &  \vdots  &            &            &     \\ 
              a_{k-1}& a_{k}       & 0           &          0   &  \ldots  &  0         &       0    & 0   \\
              a_{k}& 0            & 0           &          0   &          &  0         &       0    & 0   \\
        \end{pmatrix},
\end{equation}
where $a_1 = \int_{-\tau}^0 \phi(u)du,\ a_{k-1} = \int_{(-k+1)\tau} ^{(-k+2)  \tau} \phi(u) du, \ldots, \
a_k = \int_{-\infty}^{(-k+1)\tau} \phi(u)du.$
\subsection{Extension of the smiple ambit field simulation algorithm to unbounded ambit sets}
\label{appendix:B2}
Let $\tilde{T} = \phi^{-1}(x)+\tau,$ i.e.~the time coordinate at which the trawl function of $A_{01}$ intersects the horizontal line with spatial coordinate $x.$ Let $T=\floor{\frac{\tilde{T}}{\tau}}\tau$ be the biggest negative multiple of $\tau$ that is smaller or equal than $T;$ we break the simulation into two steps, on $\{t \le T\}$ and on $\{t > T\}.$  

We first discuss the simulation on  $\{t \le T\}.$ Note that for any $i \neq i^{'},$ we have 
\begin{equation*}
\{(t,y) \colon  i x < y < (i+1) x, t \le T\} \cap A_{i^{'}j} = \emptyset ,
\end{equation*}
and for $1 \le i \le k_s$, define 
\begin{equation}
S_{ij} = 
\begin{cases}
\{(t,y) \colon  i x < y < (i+1) x, t \le T\} \cap A_{ij} \backslash A_{i,j+1} &\text{ if } \frac{T}{\tau} <  j < k_t, \\
\{(t,y) \colon  i x < y < (i+1) x, t \le T\} \cap A_{ij} &\text{ if } j=k_t,
\end{cases}
\label{eq:s_ij_ambit_non_compact}
\end{equation}
Note that the areas of $S_{ij}$ are available in closed form, in terms of integrals of $\phi,$ as in Subsection \ref{subsection:slice_partition_algorithm}. Therefore, the sets $S_{ij}$ can be simulated exactly. We now discuss the simulation on $\{t > T\}.$

We say a set $S$ is a minimal slice with minimal pair $(k,l) = \min\{(i,j)\colon S \subset A_{ij}\}$ if
\begin{equation*}
    S = \bigcap_{(i,j) \in K} A_{ij},
\end{equation*}
for some indicator set $K \subset \{k,\ldots,k+k_s-1\} \times \{l,\ldots,l+k_t-\frac{T}{\tau}-1\},$ and $S \cap A_{i^{'}j^{'}} = \emptyset$ for any $\left(i^{'},j^{'}\right) \not\in \left(\{k,\ldots,k+k_s-1\} \times \{l,\ldots,l+k_t-\frac{T}{\tau}-1\}\right) \backslash K$; the values $l,\ldots,l+k_t-\frac{T}{\tau}-1$ come from taking into account intersections between $k_t -\frac{T}{\tau}$ consecutive ambit sets, at time coordinates $\frac{T}{\tau}+1,\ldots,k_t$ instead of just $I_t$ consecutive ambit sets at time coordinates $1,\ldots,I_t\tau$. We can still apply algorithm \ref{algo:slice_estimation_compact_ambit_set} to identify the minimal slices in the case of unbounded ambit sets, with the only difference that in step $3,$ the indicator matrices are $k_s \times \left(k_t-\frac{T}{\tau}\right)$ instead of $I_s \times I_t.$ Hence to simulate exactly $A_{ij},$ with $1 \le i \le k_s, \ 1 \le j \le k_t,$ we simulate the minimal slices in $S_{kl}$ where $\frac{T}{\tau}+1\le k \le k_t, \ -I_s+2 \le l \le k_s,$ which means that we simulate the ambit sets, or at least subsets of the ambit sets at time coordinates $\frac{T}{\tau}+1,\ldots,2k_t+\frac{T}{\tau}-1$ and space  coordinates $-I_s+2,\ldots,k_s + I_s-1$. The full procedure is given in Algorithm \ref{algo:monotonic_ambit_field_simulation_unbounded}. Similarly to subsection \ref{subsection:slice_partition_algorithm}, the computational complexity increases when we consider unbounded ambit sets. 
\input{algorithms/slice_partition_for_unbounded_ambit_sets}
\end{appendices}

%% file: algorithms/slice_partition_unbounded_monotonic_trawls.tex
\begin{algorithm}[h]
 \caption{Slice partition for bounded, monotonic trawls}\label{algo:slice_partition_trawl_process_implementation}
  \begin{algorithmic}[1]
      \Require Sampler $S(\textrm{area})$ which returns independent samples with the same law as $L(A),$ where $\mathrm{Leb}(A) = \textrm{area};$ number of trawls to be simulated $k$ and distance $\tau$ between them; $I = \ceil{-T/\tau}$.
    \Ensure Vector containing the simulated values of the trawl process at times $\tau,\ldots,k\tau.$
    \Function{main}{$S,\,k,\,\tau,\,I$}
 \State $Y \gets \text{zeros}(I,k)$
       \State Compute the areas $s_{ij}$ from  \eqref{trawl_1_areas},\eqref{trawl_2_areas}
        \For{$i = 1,\ldots,I$}
         \For{$j=1,\ldots,k$}
            \State{$Y[i,j] \gets S(s_{ij})$} \Comment{Sample a realisation of $Y_{*}$}
         \EndFor    
        \EndFor
  \State $Z=\Call{ColumnCumSum}{Y}$ \Comment{\parbox[t]{.43\linewidth}{$Z = \left(z_{ij}\right)_{\substack{1 \le i \le I \\ 1 \le j \le k}}$ with $z_{ij} = \sum_{l=i}^{I} y_{lj}$ is defined as the cumulative sum on the columns of $Y$}}  
  \State $X \gets \text{zeros}(k)$
  \For{$j=1,\ldots,k$}   
    \For{$i=1,\ldots,k$}
           \State $X[j] += Z[i,j+1-i]$ \Comment{$X_{j\tau}$ is given by the sum of entries on the $j^{th}$ diagonal of $Z$}
         \EndFor 
    \EndFor
  \Return $X$
  \EndFunction
  \end{algorithmic}
\end{algorithm}

%% file: algorithms/slice_partition_for_unbounded_ambit_sets.tex
\begin{algorithm}[h]
  \caption{Slice partition for unbounded, monotonic ambit sets}\label{algo:monotonic_ambit_field_simulation_unbounded}
  \begin{algorithmic}[1]
  \Function{SlicePartition2}{$U,\,N,\,\phi,\,T,\,I_t,\,I_s,\,\tau,x$}
 \State $I_s \gets   \ceil{\frac{\phi(0)}{x}}$
  \State $H \gets \Call{SliceEstimation}{U,\,N,\,\phi,\,I_t,\,I_s,\,\tau,\,x}$
  \State $Y \gets \text{zeros}(k_s+2I_s-2,\,2k_t-\frac{2T}{\tau} -1)$ 
  \Comment{Corresponding to the matrix  $L\left(A_{ij}\right)_{\substack{-I_s+2 \le i \le k_s+I_s-1 \\ \frac{T}{\tau}+1 \le i \le 2k_t -\frac{T}{\tau}-1 }}$}
\State  \Comment{Simulate the correction slices \eqref{eq:s_ij_ambit_non_compact}}
  \For{$k \in \{1,\ldots,k_s + I_s-1\}$}
\For{$l \in \{1,\ldots,2k_t-\frac{2T}{\tau}-1\}$}
\For {$I \in \text{keys}(H)$}
\State $c \gets T(\text{area})$\Comment{Simulate $L(S)$ for each $S \in \mathcal{S}_{kl}$}  
 \State $Y[k:k+I_s-1,\,l:l+k_t-\frac{T}{\tau}-1] \pluseq cI$ 
  \EndFor
 \EndFor
  \EndFor
  \Return $Y[I_s:I_s + k_s-1,-\frac{T}{\tau}+1:-\frac{T}{\tau}+k_t]$  \Comment{Corresponding to the matrix $L\left(A_{ij}\right)_{\substack{1 \le i \le k_s \\  1 \le j \le k_t}}$}
    \EndFunction
  \end{algorithmic}
\end{algorithm}